\documentclass[11pt]{article}  % use article class at 11pt

\usepackage{amsmath, amsthm, amssymb}
\usepackage{graphicx}
\graphicspath{ {./images/} }

\usepackage{bm} % Bold mathmode with italics retained

\usepackage[margin=1.2in]{geometry}

\usepackage{soul}
\usepackage{times,amsmath}

\newcommand{\argmin}{\operatornamewithlimits{argmin}}

\newtheorem{lemma}{Lemma}

\theoremstyle{definition}

\theoremstyle{remark}

\usepackage{verbatim}
\usepackage[toc,page]{appendix}

\newcommand{\R}{{\mathbb R}}

\newcommand{\vvec}[1]{{\bm #1}}%{\textrm{x}}

\usepackage{color}
\definecolor{gray}{gray}{0.5}

\title{\vspace{-.8in}Learning fixed points of recurrent neural networks by reparameterizing the network model}

\author{Vicky Zhu\thanks{Babson College, Mathematics, Analytics, Science, and Technology Division, Wellesley, MA, USA}\,\,\footnotemark[2]  \and Robert Rosenbaum\thanks{University of Notre Dame, Department of Applied and Computational Mathematics and Statistics, Notre Dame, IN, USA}}
\begin{document}
\maketitle

\abstract{In computational neuroscience, fixed points of recurrent neural networks are commonly used to model neural responses to static or slowly changing stimuli. These applications raise the question of how to train the weights in a recurrent neural network to minimize a loss function evaluated on fixed points. A natural approach is to use gradient descent on the Euclidean space of  synaptic weights. We show that this approach can lead to poor learning performance due, in part, to singularities that arise in the loss surface. We use a reparameterization of the recurrent network model to derive two alternative learning rules that produce more robust learning dynamics. We show that these learning rules can be interpreted as steepest descent and gradient descent, respectively, under a non-Euclidean metric on the space of recurrent weights. Our results question the common, implicit assumption that learning in the brain should be expected to follow the negative Euclidean gradient of synaptic weights. 
}

\section{Introduction}\label{Intro}

Recurrent neural network models (RNNs) are widely used in machine learning and in computational neuroscience. In machine learning, they are typically used to learn dynamical responses to time series inputs. In computational neuroscience, RNNs are sometimes used to model dynamical responses of neurons to dynamical stimuli~\cite{sussillo2009generating,sussillo2014neural}, but are also often used to model stationary, fixed point neural responses to static inputs. For example, many phenomena observed in visual cortical circuits, \emph{e.g.}, surround suppression, are widely modeled by stationary states of computational models in which recurrent connections model lateral, intralaminar connectivity~\cite{ferster2000neural,ozeki2009inhibitory,rubin2015stabilized,ebsch2018imbalanced,curto2019fixed,baker2020nonlinear}.

%Typically, in this approach, the neural response is modeled as a fixed point of the recurrent dynamics. 

A natural approach to learning fixed points of RNNs is to use direct gradient descent on the recurrent weight matrix after the network has converged toward a fixed point. 
A direct application of this approach, called ``truncated backpropagation through time,'' can be computationally expensive because it requires the application of backpropagation on a computational graph unrolled over many time steps. Moreover, backpropagation through time is difficult to implement or approximate with biologically plausible models of learning~\cite{lillicrap2019backpropagation}. 

Alternative approached use the implicit equation for fixed points to derive the exact gradients of the loss with respect to the weight matrix at the fixed point, or some approximations to this quantity~\cite{pineda1987generalization,almeida1990learning,williams1990efficient,ollivier2015training,liao2018reviving}. These approaches can also be computationally expensive and difficult to implement in biologically plausible models because the gradient derived from the implicit equation involves matrix inverses, which either need to be computed directly or approximated using, for example, iterative methods. 
In this work, we additionally show that gradient descent on the recurrent weight matrix can lead to poor learning performance because the associated loss landscape has singularities and implicit biases that make it poorly conditioned for gradient-based learning. 

When mentioning ``gradient descent'' above, we were implicitly referring to the Euclidean gradient on weights, which is standard practice. However, several authors have argued that the default use of the Euclidean gradient in gradient descent is not necessarily optimal for studying artificial or biological learning. In machine learning applications, non-Euclidean gradients informed by information theory, such as the natural gradient, are superior in some settings~\cite{amari1998natural,amari1998natural2,martens2020new}. In computational neuroscience, the use of a Euclidean gradient implicitly assumes a specific choice of units in a biological model and, more generally, assumes a specific parameterization of the model~\cite{surace2020choice,kreutzer2022natural,pogodin2023synaptic}. Different units or different parameterizations of a biological model will yield different gradients and ultimately different learning dynamics. Hence, gradient descent using the Euclidean gradient of the loss with respect to synaptic weights under a specific choice of parameterization  might not capture learning dynamics or learned representations in biological neuronal networks.

In this work, we derive two new learning rules for fixed points of recurrent neural networks by reparameterizing the network model. The first learning rule can be viewed as steepest descent with respect to a non-Euclidean metric. The second rule approximates the first one, but it is more efficient and it can be interpreted as gradient descent with a non-Euclidean gradient. We demonstrate empirically that these learning rules exhibit more robust and efficient learning dynamics than standard, Euclidean gradient descent. We also find that the parameter updates produced by these rules point in substantially different directions in parameter space than the negative Euclidean gradient. In addition to providing new, robust learning rules for learning fixed points in recurrent networks, our results question the common, implicit assumption in computational neuroscience that learning should follow the negative Euclidean gradient of synaptic weights. 

Code to apply the proposed learning rules  and produce all figures in the manuscript can be found at \\\texttt{https://github.com/RobertRosenbaum/LearningFixedPointsInRNNs}

% \vspace{.2in}
% \noindent Notation clarification:\\
% firing rates:\\
% rir_i is a scalar, a number (this should only appear in the 1D case);\\
% r\vvec{r} is a column vector of size NN;\\
% RR is a matrix of size N×mN\times m, hence RiR_i is also a vector, equivalent to r\vvec{r}, so we wrote it as ri\vvec{r}^i  \\
% Use () for vector, and [] for matrix  \\

\section{Background and theory}\label{Background}

\subsection{Model description}\label{S:model}

We consider a recurrent neural network (RNN) model of the form~\cite{dayan2001theoretical, gerstner2014neuronal,sussillo2009generating, sussillo2014neural}
\begin{equation}\label{E:drdt}
\tau\frac{d\vvec{r}}{dt}=-\vvec{r}+f(W\vvec{r}+\vvec{x})
\end{equation}
where $\vvec{r}(t)\in\R^N$ is a vector of model firing rates, $\tau>0$ is a time constant, $W\in\R^{N\times N}$ is a recurrent connectivity matrix, $\vvec{x}\in\R^N$ models external input to the network, and  $f:\R\to\R$ is a non-negative, non-decreasing activation function or ``f-I curve'', which is applied pointwise. For a time-constant input, $\vvec{x}(t)=\vvec{x}$, fixed point firing rates satisfy 
\begin{equation}\label{E:FixedPt}
\vvec{r}=f(W\vvec{r}+\vvec{x}). %= g(Wr+x)
\end{equation}
%Where gg is the a vector gain.
The stability of fixed point firing rates from Eq.~\eqref{E:drdt} is determined by the eigenvalues of the Jacobian matrix,
\begin{equation}\label{E:Jacobian}
\mathcal J =  \frac{1}{\tau}[-I+GW]
\end{equation}
where $G=\textrm{diag}(f'(\vvec{z}))$ is a diagonal matrix with entries 
\[
G_{jj}=f'(\vvec{z}_j) 
\]
and $\vvec{z}=[W\vvec{r}+\vvec{x}]$ is the vector of neural inputs or pre-activations evaluated at their fixed points. Specifically, a fixed point is hyperbolically stable if all eigenvalues of $\mathcal J$ have negative real part. 
%While stability can be achieved when some eigenvalues have zero real part, we will hereafter assume that all eigenvalues of JJ have negative real part and, as a consequence, JJ is non-singular. 

We can alternatively consider an recurrent neural network model in discrete time of the form
\begin{equation}\label{E:rn}
\vvec{r}(n+1)=f(W\vvec{r}(n)+\vvec{x}(n)).
\end{equation}
Eq.~\eqref{E:drdt} is more common in computational neuroscience while Eq.~\eqref{E:rn} is more common in machine learning, but they are closely related. 
Eq.~\eqref{E:rn} has the same fixed points as Eq.~\eqref{E:FixedPt}, but hyperbolic stability is obtained when eigenvalues of $GW$ have magnitude less than $1$. Hence, if a fixed point is stable for Eq.~\eqref{E:rn}, it is also stable for Eq.~\eqref{E:drdt}, but the converse is not true. 
In this work, we focus on the continuous system in Eq.~\eqref{E:drdt}, but  our approach and learning rules can also be  applied to the discrete system in Eq.~\eqref{E:rn}.

%\textcolor{red}{remember to move this example to later section !!! }
%\begin{exmp}
%A population with two neural subtypes E and I, we can write it as
%\[
%J=\begin{bmatrix}
%		\frac{1}{\tau_e}\big(-Id_{E}+G^eW^{ee}\big) &  \frac{1}{\tau_e}G^eW^{ei}\\
%		\frac{1}{\tau_i}G^iW^{ie} &  \frac{1}{\tau_i}\big(-Id_{I}+G^iW^{ii}\big)\\
%	\end{bmatrix}\\  %{\color{red}fill\; this\; in}
%\]
%\end{exmp}

In machine learning applications, RNNs are often used to learn mappings from input time series, $\vvec{x}(t)$, to output time series, $\vvec{r}(t)$, and they are often trained using backpropagation through time. In computational neuroscience, RNNs of the form in Eq.~\eqref{E:drdt} are often studied for their fixed point properties, for example to study orientation selectivity and surround suppression among other phenomena~\cite{ferster2000neural,ozeki2009inhibitory,rubin2015stabilized,ebsch2018imbalanced,curto2019fixed,baker2020nonlinear}, but the weights in these studies are often chosen by hand, not learned. 
As a combination of these perspectives, we are interested in {\it learning} mappings from static inputs, $\vvec{x}(t)=\vvec{x}$, to their associated fixed points, $\vvec{r}$, given by Eq.~\eqref{E:FixedPt}. 

Specifically, consider a supervised learning task with a cost function of the form
\[
J(W)=\frac{1}{m}\sum_{i=1}^m L(\vvec{r}^i,\vvec{y}^i)
\]
where $\vvec{x}^i$ is an input, $\vvec{y}^i$ is a label, $L$ is a loss function, and 
$
\vvec{r}^i=f(W{\vvec r}^i+\vvec{x}^i)
$
is the fixed point that the network converges to under input $\vvec{x}^i$.
This learning task presents unique challenges because fixed points are defined {\it implicitly} by Eq.~\eqref{E:FixedPt} instead of explicitly as a function of $\vvec{x}^i$, and also because we only wish to learn {\it stable} fixed points.  
The data set $\{(\vvec{x}^i,\vvec{y}^i)\}_i$ can be the entire data set in the case of full-batch learning, or a mini-batch in the case of stochastic learning.
Updates to $W$ during learning can be written as 
\[
W\gets W+\Delta W
\]
where
\[
\Delta W= \frac{1}{m}\sum_{i=1}^m \Delta W^i
\]
%\[
%W=W+\frac{1}{m}\sum_{i=1}^m \Delta W^i
%\]
Here, $\Delta W^i$ is an update rule that can depend on $\vvec{x}^i$, $\vvec{y}^i$, $W^i$, and $\vvec{r}^i$. %For notational convenience below, we focus on the update rule, ΔWi\Delta W^i, for a single data point and  omit ii superscript to write x\vvec{x}, y\vvec{y}, r\vvec{r} in place of  xi\vvec{x}^i, yi\vvec{y}^i, ri\vvec{r}^i, {\it etc.} 
Below, we derive and compare three different update rules, $\Delta W^i_1$, $\Delta W^i_2$, and $\Delta W^i_3$, for minimizing $J$.

%%%%%%%%%% %%%%%%%%%%%%%%%%%%%%%%%%%%%%%%%%
%%%%%%%%%%% Subsection :  Gradient descent on recurrent weight %%%%%
%%%%%%%%%%% %%%%%%%%%%%%%%%%%%%%%%%%%%%%%%%

\subsection{Gradient descent on the recurrent weight matrix.}\label{S:dW1}

%{\color{red}First set up the problem: We have inputs and labels, (xi,yi)mi=1(x_i,y_i)_{i=1}^m, and we want to do supervised learning, i.e., we want to learn a WW that minimizes  some loss function, L(ri,yi)L(r_i,y_i) where rir_i is a fixed point under xix_i. Then derive the grad descent update ΔW=−ηW∂L/∂W\Delta W=-\eta_W \partial L/\partial W}

The first learning rule we consider is direct gradient descent of the loss surface with respect to $W$ using the Euclidean gradient,
\begin{equation}\label{E:DeltaW1a}
\begin{aligned}
\Delta W^i_1=-\eta_W\nabla_W L(\vvec{r}^i,\vvec{y}^i)
%-\eta_W\nabla_W J=-\frac{\eta_W}{m}\sum_{i=1}^m \nabla_W L(\vvec{r}^i,\vvec{y}^i)
\end{aligned}
\end{equation}
where $\eta_W>0$ is a learning rate and $\nabla_W$ refers to the standard, Euclidean gradient with respect to $W$. 
If the fixed point, $\vvec{r}^i$, is hyperbolically stable, then the Jacobian matrix from Eq.~\eqref{E:Jacobian} has eigenvalues with negative real part, so  $I-G^iW=-\tau \mathcal J$ is invertible and we have (see Appendix~\ref{A:GradDesW1})

\begin{equation}\label{E:DeltaW1b}
\begin{aligned}
\Delta W_1^i
%& = -\eta_W\left(\vvec{r}^i\left(\nabla_{\vvec{r}^i}L\right)^{T} \left[I-G^iW\right]^{-1}G^i\right)^{T}.\\
= -\eta_WG^i\left[I-G^iW\right]^{-T}\left(\nabla_{\vvec{r}^i}L\right)\left(\vvec{r}^i\right)^{T}.
\end{aligned}
\end{equation}
where  $G^i=\textrm{diag}(f'(\vvec{z}^i))$ evaluated at the fixed point  and $U^{-T}$ denotes the inverse transpose of a  matrix, $U$. 
%Note that GG generally depends on ii. 
If $G^i_{jj} \neq 0$ for all $j$, then $G^i$ is  invertible so Eq.~\eqref{E:DeltaW1b} can be simplified to get
\begin{equation}\label{E:DeltaW1c}
\begin{aligned}
\Delta W_1^i
%&=  -\eta_W\left(\vvec{r}^i\left(\nabla_{\vvec{r}^i}L\right)^{T} \left[\left[G^i\right]^{-1}-W\right]^{-1}\right)^{T}.\\
=  -\eta_W
\left[\left[G^i\right]^{-1}-W\right]^{-T}\left(\nabla_{\vvec{r}^i}L\right)\left(\vvec{r}^i\right)^T
\end{aligned}
\end{equation}

Evaluating Eqs.~\eqref{E:DeltaW1b} and \eqref{E:DeltaW1c} directly is computationally expensive because they require the calculation of a matrix inverse. Truncated backpropagation through time and other methods provide alternative approaches to approximating $\Delta W_1$~\cite{pineda1987generalization,almeida1990learning,williams1990efficient,ollivier2015training,liao2018reviving}, but note that truncated backpropagation through time requires the storage of a large computational graph, making it memory inefficient. Moreover, we show in examples below that using $\Delta W_1$ to update weights can lead to poor learning performance. We next propose an alternative update rule based on a nonlinear reparameterization of the model.

\subsection{A new learning rule from reparameterizing the RNN}\label{S:dW2}

To motivate the reparameterized model,  first consider the special case of a linear network defined by
\[
f(z)=z
\]
In this case, $G=I$ is the identity matrix and 
Eq.~\eqref{E:FixedPt} for the fixed point can be written as
\[
{\bf r}=[I-W]^{-1}{\bf x}.
\]
%Gradient descent with respect to W gives Eq.~\eqref{E:DeltaW1c}. 
This is a linear model in the sense that $\bf r$ is a linear function of ${\bf x}$, but the nonlinear dependence of the cost on $W$ (especially a nonlinearity involving matrix inverses) produces complicated and computationally expensive update from Eq.~\eqref{E:DeltaW1b}.

Instead of performing gradient descent with respect to $W$, we propose instead to first apply a nonlinear change of coordinates to obtain new parameters,
\begin{equation}\label{E:AdefGeqI}
A=F(W):=[I-W]^{-1}.
\end{equation}
If we parameterize the model in terms of $A$ instead of $W$, then fixed points satisfy the standard linear model
\begin{equation}\label{E:StandardLinearModel}
{\bf r}=A{\bf x}
\end{equation}
which is linear in the input, $\bf x$, {\it and} the parameters, $A$. 
Gradient descent of the loss with respect to $A$ gives the standard update rule for a linear, single-layer neural network
\begin{equation}\label{E:DeltaAa}
\begin{aligned}
\Delta A^i &= -\eta_A\nabla_A L(\vvec{r}^i,\vvec{y}^i)\\
&=-\eta_A \left(\nabla_{\vvec{r}^i}L\right)\left(\vvec{x}^i\right)^T\\
&=-\eta_A \left(\nabla_{\vvec{r}^i}L\right)\left(\vvec{r}^i\right)^TA^{-T}
%\left[\nabla_{\vvec{r}^i}L\right]\left[\vvec{x}^i\right]^T%\\
%&=-\frac{\eta_A}{m}\sum_{i=1}^m 
%\left(\nabla_{\vvec{r}^i}L\right)(A^{-1}\vvec{r}^i)^T\\
\end{aligned}
\end{equation}
%where ∇riL\nabla_{\vvec{r}^i}L is the gradient of the loss with respect to ri\vvec{r}^i.
 where we distinguish between the learning rate, $\eta_A$, used for the reparameterized model and the learning rate, $\eta_W$, used for the original parameterization.
Eq.~\eqref{E:DeltaAa} gives a gradient-based update to the new parameter, $A$, but our original RNN model is parameterized by $W$. To update our original parameters, we need to change the $\Delta A$ from Eq.~\eqref{E:DeltaAa} back to $W$ coordinates. To do this, note that we want to find a value for $\Delta W$ that satisfies $A+\Delta A=F(W+\Delta W)$ whenever $A=F(W)$ and $\Delta A$ comes from Eq.~\eqref{E:DeltaAa}. In other words, the update to $W$ is given by 
%Solving these equations for ΔW\Delta W gives an update rule of the form
\begin{equation}\label{E:DeltaW2a}
\begin{aligned}
\Delta W_2^i & = F^{-1}(F(W)+\Delta A^i))-W\\
&=-\left[\left[I-W\right]^{-1}-\eta_A \left(\nabla_{\vvec{r}^i}L\right)\left(\vvec{r}^i\right)^T[I-W]^T\right]^{-1}+I-W
%\left[I-\left[\Delta A+ A\right]^{-1}\right]-\left[I-A^{-1}\right]\\
%& = -\left[\Delta A+A\right]^{-1}+A^{-1}\\
%& = -\frac{1}{m}\left[\left[C^0-\eta_C\left[\nabla_{\vvec{r}}L\right][\vvec{x}]^{T}\right]^{-1}- \left[C^0\right]^{-1}\right]\\
%& = -\left[-\eta_A\left[\nabla_{\vvec{r}}L\right]\vvec{x}^{T}+A\right]^{-1}+ A^{-1}\\
%& = -\eta_A\left[-\left[\nabla_{\vvec{r}}L\right]\vvec{x}^{T}+[I-W]^{-1}\right]^{-1}+ [I-W]\\
\end{aligned}
\end{equation}
where $F^{-1}(A)=I-A^{-1}$ is the inverse of $F(W)$.

To summarize this approach, if Eq.~\eqref{E:DeltaW2a} is used to update parameters, $W$, under the linear fixed point model, ${\bf r}=f(W{\bf r}+{\bf x})$ with $f(z)=z$, then the learning dynamics will be identical to standard linear regression of parameters, $A$, on the model ${\bf r}=A{\bf x}$.

Since gradient descent with respect to $A$ in Eq.~\eqref{E:DeltaAa} represents steepest descent of the loss surface in the new parameter space of $A$ and since Eq.~\eqref{E:DeltaW2a} gives the same updates in the original parameter space of $W$, 
the learning rule in Eq.~\eqref{E:DeltaW1a} corresponds to steepest descent of the cost, $J(W)$,  using a non-Euclidean metric defined by 
\begin{equation}\label{E:ddef}
d(W_1,W_2)=\|F(W_1)-F(W_2)\|
\end{equation}
where $\|B\|=\sqrt{\textrm{Tr}(BB^T)}$ is the  Euclidean or Frobenius norm on matrices. Note that $d(\cdot,\cdot)$ is a metric when restricted to the space of all matrices, $W$, for which $I-W$ is invertible.  
Hence, if we restrict to $W$ that yield hyperbolically stable fixed points, Eq.~\eqref{E:DeltaW1a} corresponds to steepest descent with respect to a non-Euclidean metric. However, the metric $d$ is not necessarily generated by an inner product, so Eq.~\eqref{E:DeltaW2a} cannot be called {\it gradient} descent since the notion of a gradient   requires a metric induced by an inner product. In Section~\ref{S:dW3}, we show that an approximation to $\Delta W_2^i$ produces gradient descent with a non-Euclidean gradient. Moreover, in Section~\ref{S:Experiments}, we present examples showing that $\Delta W_2$ is better suited to learning fixed points than the standard approach to gradient descent represented by $\Delta W_1$. But first, we need to generalize the derivation of $\Delta W_2$ to arbitrary activation functions.

% \textcolor{green}{Consider an initial value of a fixed point satisfying r0=f(W0r0+x)\vvec{r}_0=f(W_0\vvec{r}_0+\vvec{x}) and a small weight update, W=W0+ϵdWW=W_0+\epsilon dW. Then, to first order in ϵ\epsilon, the new fixed point can be written as r=r0+ϵΔr+O(ϵ2)\vvec{r}=\vvec{r}_0+\epsilon \Delta \vvec{r}+\mathcal O(\epsilon^2) where
% \[
% \Delta \vvec{r}=G[W\Delta\vvec{r}+\vvec{x}]
% \]
% This is the same }

Eq.~\eqref{E:DeltaW2a} was derived for the specific case $f(z)=z$, but we can extend it to a model with arbitrary $f(z)$. 
To do so, we first linearize Eq.~\eqref{E:FixedPt} to obtain a linearized fixed point equation,
\begin{equation}\label{E:LinFP}
\vvec{r}=G[W\vvec{r}+\vvec{x}]
\end{equation}
which has a closed form solution given by 
\begin{equation}\label{E:rAx}
{\bf r}=[I-GW]^{-1}G{\bf x}.
\end{equation}
Note, again, that $I-GW$ is invertible whenever $\vvec{r}$ is a hyperbolically stable fixed point. 

Given Eq.~\eqref{E:rAx}, a natural choice of new parameters would be 
\begin{equation}\label{E:AdefGX}
A=[I-GW]^{-1}G,
%[I-GW]^{-1}G.
\end{equation}
because it would again produce a (linearized) model of the form $\vvec{r}=A\vvec{x}$. Note that under the  linear model $f(z)=z$, we have $G=I$, and recover the parameterization in Eq.~\eqref{E:AdefGeqI}, so Eq.~\eqref{E:AdefGX} is a generalization of Eq.~\eqref{E:AdefGeqI}.  However, the update rule to $W$ derived from gradient descent on $A$ from the parameterization in Eq.~\eqref{E:AdefGX} is susceptible to blowup or singularities when some values of $G_{jj}=f'(\vvec{z}_j)$ become small in magnitude or zero. 
To see why this is the case, suppose $G_{jj}=\mathcal O(\epsilon)$ is small for some $j$ and consider an update to $W$ of the form $W=W+\Delta W$. Then the resulting update to ${\bf r}_j$ is, to linear order in $\epsilon$,
\[
\begin{aligned}
\Delta {\bf r}_j&=\sum_{k}G_{jj}\Delta W_{jk}r_k\\
&=\mathcal O(\epsilon \Delta W).
\end{aligned}
\]
On the other hand, an update of the form $A=A+\Delta A$ gives 
\[
\begin{aligned}
\Delta {\bf r}_j&=\sum_k \Delta A_{jk}{\bf r}_k\\
&=\mathcal O(\Delta A).
\end{aligned}
\]
Hence, if we want $\Delta W$ to produce the same change, $\Delta {\bf r}$, produced by $\Delta A$, then we must have $\Delta W\sim \mathcal O(\Delta A/\epsilon)$. This will cause large changes to $W$ in response to inputs for which $G$ has small elements at the fixed point, ultimately undercutting the model's performance (see Appendix~\ref{A:oldA} for more details).  In the extreme case that $G_{jj}=0$ for some $j$, updates to $W$ do not impact ${\bf r}$ ({\it i.e.}, $\Delta {\bf r}_j=0$ for any $\Delta W$ under the linear approximation ${\bf r}=G[W{\bf r}+\bf x]$), so we cannot derive a $\Delta W$ to match a given $\Delta A$, {\it i.e.}, the reparameterization in Eq.~\eqref{E:AdefGX} is ill-posed. %without accounting for nonlinear properties of $f({\bf z})$. 

To circumvent these problems, we instead take the parameterization
\begin{equation}\label{E:Adef}
A=F(W) :=[G-GWG]^{-1}%[I-GW]^{-1}G.
\end{equation}
in place of Eq.~\eqref{E:AdefGX}. 
Under the linearized fixed point equation in Eq.~\eqref{E:LinFP}, we then obtain the linear model
\begin{equation*}%\label{E:FixPtGx3}
\begin{aligned}
\bf{r}%& = GG^{-1}[I-GW]^{-1}G\bf{x}\\
%& = G[G-GWG]^{-1}G\bf{x}\\
& = GAG\bf{x}
\end{aligned}
\end{equation*}
which generalizes Eq.~\eqref{E:StandardLinearModel}. 
This equation is linear in $\vvec{x}$ and in the new parameters, $A$. Hence, learning  $A$ is again a linear regression problem, albeit with the extra $G$ terms. These extra $G$ terms prevent singularities and blowup when $G_{jj}$ terms become small or zero because $\Delta r_j=\mathcal O(\epsilon \Delta A)$ is small whenever we make an update of the form $A=A+\Delta A$ with $G_{jj}=\mathcal O(\epsilon)$ small. 
Under the simple linear model $f(z)=z$, we have $G=I$, and recover the parameterization in Eq.~\eqref{E:AdefGeqI}, so that  Eq.~\eqref{E:Adef} (like Eq.~\eqref{E:AdefGX}) is a generalization of Eq.~\eqref{E:AdefGeqI}.

Note that each input ({\it i.e.}, each $i$) will potentially have a different gain matrix, $G^i=\textrm{diag}(f'({\bf z^i}))$, so each sample will have a potentially different value of $A^{i}=[G^i-G^iWG^i]^{-1}$ as well. 
The gradient-based update of the loss, $L(\vvec{r}^i,\vvec{y}^i)$, with respect to $A^{i}$ for each sample becomes
\begin{equation*}%\label{E:DeltaAb}
\begin{aligned}
\Delta A^{i} &= -\eta_A\nabla_{A^{i}} L(\vvec{r}^i,\vvec{y}^i)\\
%&=-{\eta_A}G^{i}\left[\nabla_{\vvec{r}^i}L\right][\vvec{x}^i]^TG^{i}\\
&=-\eta_A G^{i}\left(\nabla_{\vvec{r}^i}L\right)(\vvec{r}^i)^T[G^{i}]^{-1}A^{-T}\\
\end{aligned}
\end{equation*}
%Note that, $W = [G^{i}]^{-1}\left[[G^{i}]-[A^{i}]^{-1}\right][G^{i}]^{-1}$, now we can derive a $\Delta W$ under $A$ reparameterized update,
%\textcolor{blue}{RTODO: Continue edits in eqn below.}
Using the same approach used to derive Eq.~\eqref{E:DeltaW2a} above, we can again derive an update to $W$ given by
\begin{equation}\label{E:DeltaW2b}
\begin{aligned}
\Delta W_2^i & = F^{-1}(F(W)+\Delta A^i)-W\\
&= -\left[\left[I-G^iW\right]^{-1}G^i-\eta_A\left[G^i\right]^2\left(\nabla_{\vvec{r}^i}L\right)(\vvec{r}^i)^T\left[I-G^iW\right]^TG^i\right]^{-1}\\
&\;\;\;\;\;+\left[[G^i]^{-1}-W\right].
%G^{-1}[I-GW]\\
%\left[G^i\right]^{-1}-W\\
% \Delta W_2  & = \frac{1}{m}\sum_{i = 1}^{m}\Delta W^i \\
% & = \frac{1}{m}\sum_{i = 1}^{m}\left[G^{i}\right]^{-1}\left[G^i-\left[A^{i}+\Delta A^{i}\right]^{-1}-\left[G^{i}-\left[A^{i}\right]^{-1}\right]\right]\left[G^{i}\right]^{-1}\\
%  & = -\frac{1}{m}\sum_{i = 1}^{m} \left[A^{i}-\eta_AG^{i} \left[G^{i}(\Delta A^i)G^{i}\right]^TG^{i}\right]^{-1}+\left[A^{i}\right]^{-1}\\
\end{aligned}
\end{equation}
This update can only be evaluated directly in the situation where $G_{jj}^i \neq 0$ for all $j$ so that the inverse of the gain matrix, $G$, exists. 
%In the zero gain situation, we can instead define a support domain when $G^i_{jj}>0$ and non-support domain when $G^i_{jj}=0$. 
However, note that  $[W_2^i]_{jk}\to 0$ as $G^i_{jj}\to 0$, as expected, so in situations where $G_{jj}^i=0$, it is consistent to take $[W_2^i]_{jk}=0$. 
%so the update rule could be applied by only updating entries of $W_{jk}$ for which $G_{jj}\ne 0$. %For example, when large stimulus inputs feed into a hyperbolic tangent, $f = \tanh(.)$, its derivative approaches to 0.
Note also that Eq.~\eqref{E:DeltaW2b} is equivalent to Eq.~\eqref{E:DeltaW2a} whenever $G=I$, as expected, since Eq.~\eqref{E:DeltaW2b} generalizes Eq.~\eqref{E:DeltaW2a} to the case of arbitrary $f$.

\subsection{A simpler learning rule from linearizing the reparameterized rule}\label{S:dW3}

The reparameterized rule in Eq.~\eqref{E:DeltaW2b} is rather a complicated learning rule, and the matrix inverses can be computationally expensive to compute or approximate.
%Since $\Delta W$ is associated with $\eta_A$, we can view it as a function of $\eta_A$,
If we assume that $\eta_A>0$ is small, then we can approximate Eq.~\eqref{E:DeltaW2b} by applying Taylor expansion to linear
order in $\eta_A$. This gives the linearized parameterized rule (see Appendix~\ref{A:LinearApprox} for details),
\begin{equation}\label{E:DeltaW3}
\begin{aligned}
\Delta W_3^i 
%& = -\eta_A[A^i]^{-1}[G^i]^2\left[\nabla_{\vvec{r}^i}L\right]\left[\vvec{x}^i\right]^{T}[G^i]^2[A_i]^{-1}\\
& = -\eta_A\left[I-W G^i\right]G^i\left(\nabla_{\vvec{r}^i}L\right)(\vvec{r}^i)^T\left[I - G^i W\right]^T[I-G^iW]\\
%&=-\frac{\eta_A}{m}\sum\limits_{i=1}^{m}\left[I-WG^i\right]G^{i}\left(\nabla_{\vvec{r}^i}L\right)\left(\vvec{x}^i\right)^TG^{i}[I-G^iW]\\
\end{aligned}
\end{equation}
%The radius of convergence of the associated Taylor polynomial is given by 
%\begin{equation}\label{E:radius}
%\begin{aligned}
%\left\|\eta_A \left[\left[A^{i}\right]^{-1}\right]\left[G^{i}\right]^2\left(\Delta A^{i}\right)\left[G^i\right]^2\right\|& < 1.\\
%%\|\eta_A \left[G^i-W\right][G^{i}]^2\left(\nabla_{\vvec{r}^i}L \right)\left(\vvec{r}^{i}\right)^T\left[G^{i}-G^{i}W^{T} G^{i}\right][G^{i}]^2\|& < 1
%\end{aligned}
%\end{equation}
%Therefore, if
%\[
%\eta_A \ll \frac{1}{\|\left[G^{i}-G^{i}W G^{i}\right][G^{i}]^2\left(\nabla_{\vvec{r}^i}L \right)\left(\vvec{r}^{i}\right)^{T}\left[G^i-G^iW^TG^i\right][G^{i}]^2\|},
%\]
%the $\Delta W_3^i$ from Eq.~\eqref{E:DeltaW3} approximates $\Delta W^i_2$ from Eq.~\eqref{E:DeltaW2b}.
In contrast to Eqs.~\eqref{E:DeltaW1b} and \eqref{E:DeltaW2b} for $\Delta W_1^i$ and $\Delta W_2^i$, Eq.~\eqref{E:DeltaW3} for $\Delta W_3^i$ does not require the computation of matrix inverses. Like $\Delta W_1^i$ and $\Delta W_2^i$, $\Delta W_3^i$ satisfies $\Delta W_{jk}\to 0$ whenever $G_{jj}\to 0$,  but unlike Eq.~\eqref{E:DeltaW2b} for $\Delta W_2^i$, Eq.~\eqref{E:DeltaW3} for $\Delta W_3^i$ can be evaluated directly when $G_{jj}=0$ for some $j$.

%Recall that $\Delta W_2^i$ can be interpreted as steepest descent with respect to a non-Euclidean metric, but it cannot be interpreted as gradient descent.  
%because the corresponding metric does not generate an inner product, which is necessary for defining a non-Euclidean gradient.
Notably, $\Delta W_3^i$ can be interpreted as gradient descent of the loss function with a non-Euclidean gradient. To see why this is the case, first note that $\Delta W^i_3$ is related to $\Delta W^i_1$ according to  
\begin{equation}\label{E:AtoEmetric}
\begin{aligned}
\Delta W^i_3 = B^i \Delta W^i_1 C^i,
\end{aligned}
\end{equation}
where 
\[
B^i =[I-WG^i][I-WG^i]^T
\]
and 
\[
C^i = [I-G^iW]^T[I-G^iW].
\]
Here and for the remainder of this section, we take $\eta_A=\eta_W=\eta$ to highlight the relationship between the two update rules, but constant scalar coefficients do not affect these results. 

Using Eq.~\eqref{E:AtoEmetric}, we may conclude that $\Delta W_3^i$ is equivalent to gradient descent of the loss with respect to $W$ using a non-Euclidean gradient. 
To explain this statement in more detail, note that the gradient of $L({\bf r}^i,{\bf y}^i)$ with respect to $W$ depends on the choice of metric or geometry~\cite{surace2020choice}. Given an inner product, $\langle \cdot,\cdot\rangle_a$, on $\R^{N\times N}$ the gradient of a function, $F:\R^{N\times N}\to \R$, on the geometry imposed by $\langle \cdot,\cdot\rangle_a$ evaluated at $W\in \R^{N\times N}$ is the unique matrix $\nabla^a_W F\in \R^{N\times N}$ such that for every $U\in \R^{N\times N}$~\cite{spivak2018calculus,surace2020choice}, 
\[
\langle \nabla^a_W F, U\rangle_a=\lim_{\epsilon\to 0}\frac{F(W+\epsilon U)-F(W)}{\epsilon}.
\]
The standard Euclidean gradient, $\nabla=\nabla^E$, on matrices is given by taking the geometry produced by the Euclidean or Frobenius inner product, 
\[
\langle U,V\rangle_E=\sum_{jk}U_{jk}V_{jk}=\textrm{Tr}(UV^T).
\]
Recall that $\Delta W_1^i$ is defined by the Euclidean gradient,
\[
\Delta W_1^i=-\eta \nabla_{W}^E L({\bf r}^i,{\bf y}^i)
\]
where $L({\bf r^i},{\bf y}^i)$ is interpreted as a function of $W$. We claim that 
\begin{equation}\label{E:DW3nablaB}
\Delta W_3^i=-\eta  \nabla_{W}^{B} L({\bf r}^i,{\bf y}^i)
\end{equation}
where $\nabla_{W}^B$ is the gradient under the geometry defined by the inner product,
\[
\begin{aligned}
\langle U,V\rangle_B&=\textrm{Tr}(B^{-1}UC^{-1}V^T)\\
&=\langle B^{-1}U,VC^{-1}\rangle_E.
\end{aligned}
\]
Note that we can use the cyclic property of the trace operator to write
\[
\begin{aligned}
\langle U,V\rangle_B&=\textrm{Tr}(B^{-1}UC^{-1}V^T)\\
&=\textrm{Tr}\left([I-WG]^{-T}[I-WG]^{-1}U[I-GW]^{-1}[I-GW]^{-T}V^T\right)\\
&=\textrm{Tr}\left([I-WG]^{-1}U[I-GW]^{-1}[I-GW]^{-T}V^T[I-WG]^{-T}\right)\\
&=\langle \mathcal LU,\mathcal L V\rangle_E
\end{aligned}
\]
where $\mathcal L:\R^{N\times N}\to \R^{N\times N}$ is a linear operator on $N\times N$ matrices defined by
\[
\mathcal L(U)=[I-WG]^{-1}U[I-GW]^{-1}.
\]
Hence, $\langle \cdot,\cdot\rangle_B$ can be viewed as a Euclidean inner product on linearly transformed coordinates. This confirms that $\langle \cdot,\cdot\rangle_B$ defines an inner product on square matrices whenever $[I-WG]$ and $[I-GW]$ are non-singular. 
%See Appendix~\ref{A:InnerProduct} for a proof that $\langle U,V\rangle_B$ defines an inner product on matrices whenever $I-WG$ and $I-GW$ are non-singular. 
For notational convenience here and below,  we do not write the explicit dependence of $B$, $C$, or $\mathcal L$ on $i$, but  they {\it do} depend on $i$ through $G^i$. In other words, there are distinct matrices, $B$ and $C$, and therefore distinct inner products, $\langle\cdot,\cdot\rangle_B$, at each gradient descent iteration.
Given Eq.~\eqref{E:AtoEmetric}, we can prove Eq.~\eqref{E:DW3nablaB} by showing that
\begin{equation}\label{E:nalbaBE}
\nabla^{B}_{W} L=B\left[ \nabla^{E}_{W} L\right]C.
\end{equation}
To show Eq.~\eqref{E:nalbaBE}, first define the $N\times N$ standard basis matrices ${\vvec 1}^{jk}\in \R^{N\times N}$ entrywise by
\[
{\vvec 1}^{jk}_{j'k'}=
\begin{cases}
1 & j=j'\textrm{ and }k=k'\\
0 & \textrm{otherwise}
\end{cases}.
\]
for $j,k=1,\ldots,N$.
Now compute the inner product of the gradient with ${\vvec 1}^{jk}$,
\begin{equation}\label{E:EB1}
\begin{aligned}
\left\langle\left[\nabla^{B} L\right],{\vvec 1}^{jk}\right\rangle_{B}&=\left\langle B^{-1}\left[\nabla^{B} L\right], {\vvec 1}^{jk}C^{-1}\right\rangle_E\\
&=\textrm{Tr}\left(B^{-1}\left[\nabla^{B}L\right]C^{-1}\left[{\vvec 1}^{jk}\right]^T\right)\\
&=\sum_{n=1}^N \left[B^{-1}\left[\nabla^{B} L\right]C^{-1}{\bf 1}^{kj}\right]_{n,n}\\
&=\sum_{n,m=1}^N \left[B^{-1}\left[\nabla^{B} L\right]C^{-1}\right]_{n,m} [{\bf 1}^{kj}]_{m,n}\\
&=\left[B^{-1}\left[\nabla^{B} L\right]C^{-1}\right]_{jk}
\end{aligned}
\end{equation}
where the last line follows from the fact that ${\bf 1}^{kj}_{n,m}=1$ when $n=k$ and $m=j$, and it is equal to zero for all other $j,k$. 
But we also have, from the definition of a gradient,
\begin{equation}\label{E:EB2}
\left\langle\left[\nabla^{B} L\right],{\vvec 1}^{jk}\right\rangle_{B}=\lim_{\epsilon\to 0}\frac{J(W+\epsilon {\vvec 1}^{jk})-J(W)}{\epsilon}=\frac{\partial J}{\partial W_{jk}}=\left[\nabla^{E} L\right]_{jk}.
\end{equation}
Since Eqs.~\eqref{E:EB1} and \eqref{E:EB2} apply for all indices $j,k$, we may conclude that 
\[
B^{-1}\left[\nabla^{B} L\right]C^{-1}=\left[\nabla^{E} L\right]
\]
and therefore that
\[
\left[\nabla^{B} L\right]=B\left[\nabla^{E} L\right]C,
\]
which concludes our proof.

In summary, if $W$ is updated according to $\Delta W^i_3$ from Eq.~\eqref{E:DeltaW3}, then this is equivalent to performing gradient descent on the loss with respect to the weight matrix under the geometry defined by the new inner product, $\langle U,V\rangle_{B}$.  Below, we present examples showing that this geometry is better suited to learning $W$ than gradient descent with respect to the standard Euclidean geometry. Specifically, $\Delta W_3^i$ learns more robustly than $\Delta W_1^i$.

\section{Experiments and results}\label{S:Experiments}

We next evaluate and interpret each of the learning rules derived above on two different learning tasks.

%%%%%%%%%% %%%%%%%%%%%%%%%%%%%%%%
%%%%%%%%%%% Subsection : Linear Least Squares in Higher-dim %%%%%%
%%%%%%%%%%% %%%%%%%%%%%%%%%%%%%%%

\subsection{Learning fixed points in a linear model.}\label{S:linearND}

% We next generalize the results from the previous section to higher dimensions ($N>1$). 

For demonstrative purposes, we first consider an example of linear regression with mean-squared loss. 
Specifically, we consider $f(z)=z$ with
\[
L(\vvec r,\vvec y)=\|\vvec r-\vvec y\|^2.
\]
where $\|\cdot\|$ is the Euclidean norm on $\R^N$. 
Note that $G=I$ is the identity in this case. 
We define the $N\times m$ matrices, $X=\left[{\bf x}^1\;\; {\bf x}^2\;\ldots {\bf x}^m\right]$, $Y=\left[{\bf y}^1\;\; {\bf y}^2\;\ldots {\bf y}^m\right]$, and $R=\left[{\bf r}^1\;\; {\bf r}^2\;\ldots {\bf r}^m\right]=[I-W]^{-1}X$. 
The cost function can  be written as
\begin{equation}\label{E:Wcost}
J(W)=\frac{1}{m}\left\|\,[I-W]^{-1}X-Y\,\right\|^2.
\end{equation}
It is useful to also write the cost in terms of the parameters $A=[I-W]^{-1}$ to get a standard quadratic cost function,
\begin{equation}\label{E:Acost}
J_A(A)=\frac{1}{m}\left\|AX-Y\right\|^2.
\end{equation}
For this problem, minimizers of $J(W)$ and $J_A(A)$ can be found explicitly. Before continuing to empirical examples, we derive and discuss these explicit minimizers.

\subsubsection{Computing explicit minimizers in a linear model.}

In the under-parameterized case ($N\le m$ when all matrices full rank), $J_A(A)$ has a unique minimizer defined by 
\[
A^*=YX^+
\]
where $X^+=X^T(XX^T)^{-1}$ is the Moore-Penrose pseudo-inverse of $X$ when $N\le m$. Therefore, $J(W)$ has a unique minimizer at
\[
W^*=I-[A^*]^{-1}=XX^T(YX^T)^{-1}
\]
under the assumption that $A^*$ is invertible.

The over-parameterized case ($N>m$ when matrices are full rank) is more relevant and interesting. In this case, there are infinitely many choices of $W$ and $A$ for which $J(W)=0$ and $J_A(A)=0$. 
The problem of choosing a solution to $J_A(A)=0$ is a standard least squares problem and a common approach is to take 
\[
A^*=YX^+
\] 
where $X^+=(X^TX)^{-1}X^T$ is the Moore-Penrose pseudo-inverse of $X$ when $N>m$. It is tempting to use this value of $A^*$ and then take $W^*=I-\left[A^*\right]^{-1}$. However, note that $A^*$ is the solution to $AX=Y$ that minimizes the Frobenius norm of $A$, {\it i.e.},
\[
A^*=\argmin_A \|A\| \;\textrm{ s.t. }\; AX=Y.
\]
Therefore, $W^*=I-\left[A^*\right]^{-1}$  represents a solution, $W$, that minimizes the norm of $A=[I-W]^{-1}$. Since the Jacobian matrix is given by $\mathcal J=(-I+W)/\tau=-A^{-1}/\tau$, stability is promoted by $W$ having a small spectral radius (all eigenvalues of $W$ must have real part less than $1$ for stability). Hence, $W^*=I-\left[A^*\right]^{-1}$ is a poor choice for $W^*$. Minimizing the Frobenius norm of $A$ will tend to push the eigenvalues of $A$ toward zero, which can lead to large eigenvalues of $W=I-A^{-1}$ and $\mathcal J=-A^{-1}/\tau$, promoting unstable fixed points. 
Instead, to find a good optimizer, $W^*$, we can find solutions that minimize the norm of $W$ instead of $A$. To this end, we can solve
\[
W^*=\argmin_W \|W\| \;\textrm{ s.t. }\; [I-W]^{-1}X=Y.
\]
To solve this problem, we re-write it in a more standard form
\[
W^*=\argmin_W \|W\| \;\textrm{ s.t. }\; WY=Y-X.
\]
This problem has the solution 
\begin{equation}\label{E:Wstar}
W^*=[Y-X]Y^+
\end{equation}
where $Y^+=(Y^TY)^{-1}Y^T$ is the Moore-Penrose pseudo-inverse of $Y$ when $N>m$. This  is the solution with minimal Frobenius norm on $W$ and is therefore more likely than $I-\left[A^*\right]^{-1}$ to have a small spectral radius and therefore more likely to give stable fixed points. Hence, Eq.~\eqref{E:Wstar} provides a good optimizer in the over-parameterized case ($N>m$).

%We used the same cost function in Eq.~\eqref{E:LinearEX}, where now the rate and response vector become a matrix instead $R,Y\in \R^{N\times m}$. For  $ i = 1, \dots, m$, $\vvec{y}^i\in \R^N$ is the $i$th target. We again use $f(z)=z$ so that $G = I$ and  fixed point firing rates are given by ${\bf r}=[I-W]^{-1}{\bf x}$.

\subsubsection{Visualizing the loss landscape of a linear model.}

For empirical examples, we first generated inputs, ${\bf x}^i$, independently from a Gaussian distribution and generated targets ${\bf y}^i$ using a ground truth weight matrix, $\hat W$, and adding noise. Specifically, we define
\begin{equation}\label{E:setupNd}
\begin{aligned}
X & \sim \sigma_x  Z_{N\times m}\\
Y & \sim \left[I - \hat W\right]^{-1}X+\sigma_y Z_{N\times m}
\end{aligned}
\end{equation}
where $\sigma_x = 0.1$ controls the magnitude of the inputs, $\sigma_y = 0.01 $, and each $Z_{N\times m}$ represents an $N\times m$ matrix of independent, standard, Gaussian random variables. The ground truth weight matrix is generated by
\begin{equation*}%\label{E:Wborn}
\begin{aligned}
\hat W & \sim  \frac{\sigma_w}{\sqrt{N}}Z_{N\times N}.\\
\end{aligned}
\end{equation*}
Following Girko's circular law, the eigenvalues of $\hat W$ lie approximately within a circle of radius $\sigma_w$ with high probability~\cite{girko1985circular}. Hence, we take  $\sigma_w=0.5 <1 $ to control the spectral radius of the circle to be less than $1$, so that all eigenvalues of the Jacobian matrix, $\mathcal J = (-I+\hat W)/\tau$, are negative and fixed point firing rates are stable under the ground truth parameters, $\hat W$.

% \textcolor{blue}{The paragraph below won't make any sense to the reader right now. It is not even very clear to me right now. See my attempt to explain this figure below.}
% σw\sigma_w further defines a stability boundary (Figure 3\vvec{A,B}, blue dash line for the 1D and 2D visualization respectively) as
% \begin{equation}\label{E:stability}
% \begin{aligned}
% \text{1d visualization: } w & = \pm \frac{\sqrt{1-\sigma_w^2}}{\sigma_w p} \\
% \text{2d visualization: } x^2+y^2 & = \frac{1-\sigma_w}{\sigma_w^2p^2},\\
% \end{aligned}
% \end{equation}
% where p is the perturbation scale for generating some 1D visualisation paths that are different from our initial connectivity matrix WW (Figure 3A: different color curves).

%\textcolor{blue}{I moved the description to here.}

The cost landscape, $J(W)$, cannot easily be visualized as a function of $W$ for $N>1$ because $W$ has $N^2$ dimensions, so even $N=2$ would be difficult to visualize. To help visualize $J(W)$, we first plotted it on a random line segment passing through $W^*$ in $\R^{N\times N}$. Specifically, we defined the parameterized function
\begin{equation}\label{E:Wt}
W(t)=W^*+\frac{ct}{\sqrt N}Z_{N\times N}
%(1-t)W^*+t\sigma_? Z\ldots {\color{blue} \textrm{some function of tt, W∗ W^* (or ˆW\hat W?), and some random matrix that we can call ZZ}}
\end{equation}
where $c=2.5$ scales the maximum magnitude of the perturbation and $t$ was varied from $-1$ to $1$ to create the visualization of $J(W(t))$ (Figure~\ref{LossWAnD}A).  
This corresponds to plotting $J(W)$ along a one-dimensional slice of the space $\R^{N\times N}$ on which $W$ lives. Note that the true minimizer, $W=W^*$, is sampled when $t=0$.  The values of $W$ sampled by the process can produce stable or unstable fixed points. Making the approximation $W^*\approx \hat W$, we have that $\rho(W(t))\approx \sqrt{\sigma_w^2+c^2t^2}$ and therefore an approximate stability condition is given by $|t|<\sqrt{1-\sigma_w^2}/{c}\approx 0.346$.

Figure~\ref{LossWAnD}A shows the resulting cost curve for five random values of $Z$ with the blue dashed lines marking the approximate stability boundary. The cost is relatively well behaved within the boundary,  but poorly conditioned outside of the boundary because of the singularities produced by the matrix inverses in Eq.~\eqref{E:Wcost}. Specifically, outside of the stability region, the spectral radius of $W$ is larger than $1$ so some eigenvalues are near $1$ in magnitude. As a result, the $[I-W]^{-1}$ in Eq.~\eqref{E:Wcost} can lead to very large values of $J(W)$.

%%%%%%%%%%%%%%%%%%%%%%%%%% Figure 4 %%%%%%%%%%%%%%%%%%%%%%%%%%%%%%
 \begin{figure}
 \centering{
 \includegraphics[width=6in]{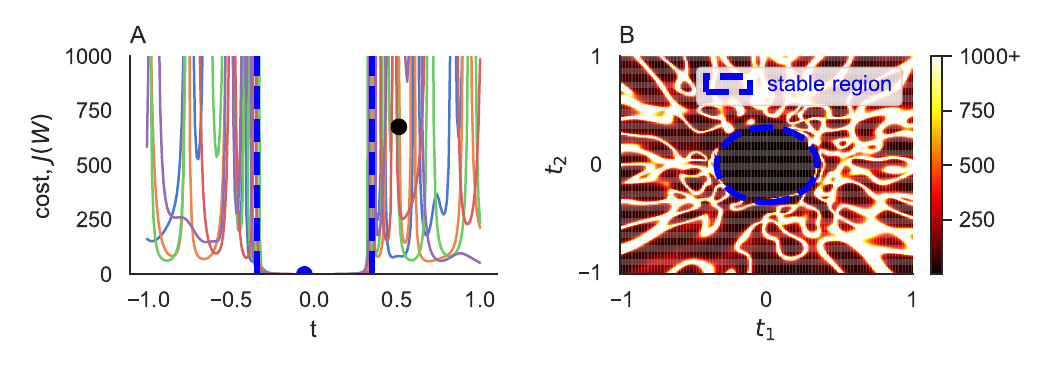}
 }
 \caption{{\bf Visualizing the cost landscape for a linear model.} {\bf A)} The cost function $J(W(t))$ as a function of $t$ from Eq.~\eqref{E:Wt}. This represents the cost evaluated along five random line segments in $\R^{N\times N}$, each passing through $W^*$ at $t=0$.  Two blue dash lines show the stability boundary, $|t|=0.346$.  The vertical axis is cutoff at $J=1000$ to better visualize the curves. Blue and black circles show stable and unstable initial conditions used  for learning. 
 {\bf B)} The cost function $J(W(t_1,t_2))$ from Eq.~\eqref{E:Wt1t2}. This represents the cost evaluated on a randomly oriented square with center at $W^*$. The color axis is cutoff at $J=1000$. 
%{\bf B)} Loss curves in the log scale, which gives a closer look in more details.  Although the instability regions looks very different from stability part, (the blue dot represent points within the stability region, whereas the black dot represents points from the instability region.). {\bf C, D)} A loss landscape using heatmap in (x,y)-coordinate from -1 to 1. Colour represents different loss values: bigger loss corresponds to a brighter colour (similar to A, C is in the normal scale with a loss value limited within 10310^3 \textit{versus} D is in the log scale without imposing any limitations to the loss values).
}
 \label{LossWAnD}
 \end{figure}
%%%%%%%%%%%%%%%%%%%%%%%%%% End of Figure 4

To further visualize the loss landscape, we repeated the procedure above in two dimensions by sampling values of $W$ from a random {\it plane} passing through $W^*$. Specifically, we defined the parameterized function
\begin{equation}\label{E:Wt1t2}
W(t_1,t_2)=W^*+\frac{c}{\sqrt N}(t_1 Z_1+t_2Z_2)
\end{equation}
where $Z_1,Z_2\sim N_{N\times N}(0,1)$, $t_1$ and $t_2$ were each varied from $-1$ to $1$ to create the visualization of $J(W(t))$, and $c=2.5$ scales the perturbation (Figure~\ref{LossWAnD}B).  
Note that $W(t_1,t_2)=W^*$ when $t_1=t_2=0$, so the center of the square corresponds to the minimum cost, $J=0$. The approximate stability condition becomes $\sqrt{t_1^2+t_2^2}<\sqrt{1-\sigma_w^2}/{c}=0.346$, so the approximate stability boundary is a circle (Figure~\ref{LossWAnD}B, dashed blue curve). Singularities create intricate ridges of large cost outside of the stability boundary (Figure~\ref{LossWAnD}B). 

In summary, Figure~\ref{LossWAnD} shows that the cost landscape, $J(W)$,  is extremely poorly conditioned outside of the stability region, {\it i.e.}, when $W$ has a spectral radius larger than $1$. Note, however, that the effective cost landscape, $J_A(A)$, of the reparameterized model is a simple quadratic landscape, given by Eq.~\eqref{E:Acost}. Gradient-based learning according to $\Delta W_1$ must traverse the poorly conditioned landscape from Figure~\ref{LossWAnD}. But the learning dynamics of the reparameterized rule, $\Delta W_2$, are equivalent to those produced by $A$ traversing a comparatively well-behaved quadratic landscape. We show in empirical examples below that this difference helps $\Delta W_2$ and its linear approximation, $\Delta W_3$, perform more robustly than $\Delta W_1$. 

%On the other hand, note that $J_A(A)$ is a simple quadratic potential, which does not have any singularities. Hence, we expect learning to be 

%Note that this expression solves for both underparameterized and overparameterized systems since it's a direct optimization problem with respected to $W$. Compared to Eq.~\eqref{E:Afixed} approach, directly solving the optimal value of $W^*$ promotes stability and avoid the floating point errors from the calculation of the optimizer $A^*$ inverse.\\

%Moreover, we can obtain the optimal $W^* = [\vvec{y}-\vvec{x}]\vvec{y}^{+}$ within stability region from Eq.~\eqref{E:optimalW} (\vvec{see Appedix for the proof}). 

%%%%%%%%%%%%%%%%%%%%%%%%%%%%%%%%%%%%%%%%%%%%%%%%%%%%%%%%%%%%%%%%%%%%%%%
%%%%%%%%%%%%%%%%%%%%%%%%%%%%%%%%% sub 2.4.1 %%%%%%%%%%%%%%%%%%%%%%%%%%%%%%%%%
\subsubsection{Gradient descent on the recurrent weight matrix in a linear model.}\label{S:linearNDdW1}

We first perform direct gradient descent on $J$ with respect to $W$ using $\Delta W_1$. The gradient-based update rule from Eq.~\eqref{E:DeltaW1c} can be written as
\begin{equation}\label{E:DeltaW1Lin}
\begin{aligned}
\Delta W_1&=\frac{1}{m}\sum_{i=1}^m \Delta W_1^i\\
%&=-\eta_W\frac{dJ}{dW}\\
%& = -\eta_W\frac{2}{m}\sum\limits_{i=1}^m \left[\vvec{r}_i(\vvec{r}_i-\vvec{y}_i)^{T}\left[I-W\right]^{-1}\right]^{T}.\\
%& = -\eta_W\frac{2}{m} \left[R(R-Y)^{T}\left[I-W\right]^{-1}\right]^{T}.\\
& = \frac{-2\eta_W}{m}
\left[I-W^{T}\right]^{-1}\left[R-Y\right]R^T.
\end{aligned}
\end{equation}
Empirical simulations show relatively poor learning performance (Figure~\ref{Over3dWComp}A). Learning is slow for small learning rates, but larger learning rates fail to converge to good minima. Recall that the true minimum is zero because the model is over-parameterized. We next show that the linearized approximation to $\Delta W_2$ performs similarly well. 

%Following the discussion above, we hypothesized that gradient descent using initial values of $W$ from the stable region (spectral radius $<1$) would perform better compared to the initial values from unstable region (spectral radius $>1$).
%
%Indeed, when the initial $W$ was chosen from within the stability region, gradient descent approached the minimum cost, $J=0$, after sufficiently long learning for sufficiently small learning rates (Figure~\ref{Over3dWComp}A). In contrast, when the initial $W$ was chosen from outside the stability region, the cost remains high across a wide range of learning rates (Figure~\ref{Over3dWComp}B). 
%We next show that the reparameterized learning rule, $\Delta W_2^i$, from Eqs.~\eqref{E:DeltaW2a} and \eqref{E:DeltaW2b} improves  robustness to the initialization of $W$. 

%Even for initial conditions within the stability region (Figure~\ref{Over3dWComp}A), the effectiveness of gradient-based learning was somewhat sensitive to the choice of learning rate, consistent with the observations we made for one-dimensional linear regression. We next show that the reparameterized learning rule, $\Delta W_2^i$, from Eqs.~\eqref{E:DeltaW2a} and \eqref{E:DeltaW2b} improves learning robustness. 

%%%%%%%%%%%%%%%%%%%%%%%%%%%%%%%%% End of sub 2.4.1 %%%%%%%%%%%%%%%%%%%%%%%%%%%%%
%%%%%%%%%%%%%%%%%%%%%%%%%%%%%%%%%%%%%%%%%%%%%%%%%%%%%%%%%%%%%%%%%%%%%%%

%%%%%%%%%%%%%%%%%%%%%%%%%% Figure 5 %%%%%%%%%%%%%%%%%%%%%%%%%%%%%%
 \begin{figure}
 \centering{
 \includegraphics[width=6in]{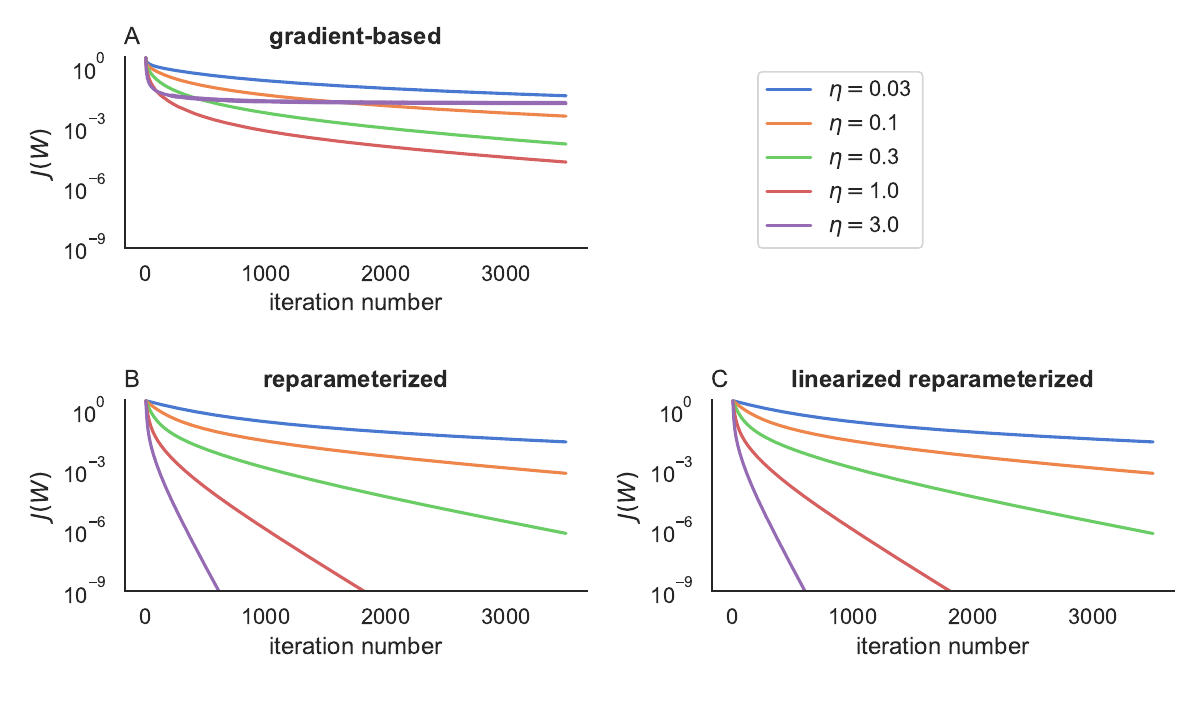}
 }
 \caption{{\bf Performance of three different learning rules for a linear regression problem.} 
 {\bf A)} The cost function, $J(W)$, from Eq.~\eqref{E:Wcost} for five different learning rates, $\eta=\eta_W$, using direct gradient descent on the weight matrix, $\Delta W_1^i$ from Eq.~\eqref{E:DeltaW1Lin}. {\bf B)} Same as A, but for the reparameterized learning rule, $\Delta W_2^i$ from Eq.~\eqref{E:DeltaW2Lin} with $\eta=\eta_A$. {\bf C)} Same as B, but for the linearized rule, $\Delta W_3^i$ from Eq.~\eqref{E:DeltaW3Lin}.
% Three different learning rules: $\Delta W_1^i$ (top row; A,B), $\Delta W_2^i$ (middle row; C,D), and $\Delta W_3^i$ (bottom row; E,F) applied to a linear regression problem in $N=200$ dimensions with $m=100$ data points, so the system is over-parameterized ($N>m$). In the left column (A,C,E) the initial $W$ was drawn from within the stability region (spectral radius $<1$; blue dot in Figure~\ref{LossWAnD}A). In the right column (B,D,F) the initial $W$ was drawn from outside of the stability region (spectral radius $>1$; black dot in Figure~\ref{LossWAnD}A). 
 }
 \label{Over3dWComp}
 \end{figure}
%%%%%%%%%%%%%%%%%%%%%%%%%% End of Figure 5

%%%%%%%%%%%%%%%%%%%%%%%%%%%%%%%%%%%%%%%%%%%%%%%%%%%%%%%%%%%%%%%%%%%%%%%
%%%%%%%%%%%%%%%%%%%%%%%%%%%%%%%%% sub 3.2.2: Reparameterization %%%%%%%%%%%%%%%%%%%%%%%%%%%%%%%%%
\subsubsection{Learning using the reparameterized rule in a linear model.}\label{S:linearNDdW2}

For this linear example, the reparameterized learning rule from Eqs.~\eqref{E:DeltaW2a} and \eqref{E:DeltaW2b} can be written as
\begin{equation}\label{E:DeltaW2Lin}
\begin{aligned}
\Delta W_2%&=\frac{1}{m}\sum_{i=1}^m\Delta W_2^i\\
%&=A^{-1}-(A+\Delta A)^{-1}\\
%&={I-W-\left(\left[I-W\right]^{-1}-\eta_A \frac{2}{m}\sum\limits_{i=1}^m (\vvec{r}_i-\vvec{y}_i)\vvec{x}_i^T\right)^{-1}}
&=[I-W]-\left(\left[I-W\right]^{-1}-\frac{2\eta_A}{m} (R-Y)X^T\right)^{-1}.
\end{aligned}
\end{equation}
Recall that the learning dynamics produced by Eq.~\eqref{E:DeltaW2Lin} are equivalent to those produced by learning the standard quadratic cost function, $J_A(A)$, 
%for a linear model
%\begin{equation}\label{E:Acost}
%\begin{aligned}
%J_A(A) & = \frac{1}{m} \sum\limits_{i = 1}^{m}\big(A\vvec{x}^i-\vvec{y}^i\big)^T\big(A\vvec{x}^i-\vvec{y}^i\big)
%\end{aligned}
%\end{equation}
along with the standard gradient-based update rule,
\begin{equation}\label{E:DeltaALinear2}
\begin{aligned}
\Delta A
%&=-\eta_A \frac{dJ}{dA}\\
%&= -\eta_A \frac{2}{m} \sum\limits_{i=1}^m (A\vvec{x}_i-\vvec{y}_i)\vvec{x}_i^T\\
&= -\frac{2\eta_A }{m} (AX-Y)X^T,
\end{aligned}
\end{equation}
which is often called the ``delta rule.''

The behavior of the learning dynamics under Eq.~\eqref{E:DeltaALinear2} in the overparameterized case is well understood~\cite{gunasekar2017implicit,soudry2018implicit,gunasekar2018characterizing,zhang2021understanding}. Specifically, $A$ tends toward solutions to $AX=Y$ that minimize the distance, $\|A-A_0\|$, of $A$ from its initial condition under the Frobenius norm. As a result, $\Delta W_2$ finds solutions, $W$, to $[I-W]^{-1}X=Y$ that minimize the distance, $d(W,W_0)$, of $W$ from its initial condition under the metric, $d$, defined in Eq.~\eqref{E:ddef}. In addition, since the Jacobian matrix is given by $\mathcal J=-A^{-1}/\tau$, we may conclude that $\Delta W_2$ finds solutions that minimize the distance, $\left\|\mathcal J_0^{-1}-\mathcal J^{-1}\right\|$, between the inverse Jacobian and its initial condition under the Frobenius norm. 

%In the limit of a large number of learning iterations, $A$ converges to the solution to $AX=Y$ that minimizes the distance between $A$ and the initial parameters, $A_0$~\cite{}. In other words,
%\[
%A^\infty = \argmin_{A} \left\|A_0-A\right\|  \;\textrm{ s.t. }\; AX=Y
%\] 
%where $A_0=I-W_0^{-1}$, $W_0$ is the initial weight matrix, and $\|\cdot\|$ is Frobenius norm. Therefore, under $\Delta W_2$, $W$ converges to the solution
%\[
%W^\infty=\argmin_{W} d(W_0,W)   \;\textrm{ s.t. }\; [I-W]^{-1}X=Y
%\]
%where $d(W_0,W)$ is the metric defined in Eq.~\eqref{E:ddef}. %Hence, learning through $\Delta W_2$ avoids singularities because the norm above diverges as $W$ approaches a singular matrix. 
%%Singularities in the cost landscape occur when $I-W$ becomes singular (Figure~\ref{}). But note that $\left\|[I-W^0]^{-1}-[I-W]^{-1}\right\| \to \infty$ as $W$ approaches a singular matrix. Therefore, $W^\infty$ will never occur at a singularity. 
%In addition, since the Jacobian matrix is given by $J=-A^{-1}/\tau$, the Jacobian matrix converges to
%\[
%J^\infty=\argmin_{J} \left\|J_0^{-1}-J^{-1}\right\|   \;\textrm{ s.t. }\; \tau J^{-1}X=-Y
%\]
%where $J_0$ is the initial Jacobian matrix. Therefore, learning with $\Delta W_2$ converges the solution that minimizes the distance between the inverse Jacobian matrix and its initial value. 

In comparison to the gradient-based update, $\Delta W_1$,  from Eq.~\eqref{E:DeltaW1Lin} (Figure~\ref{Over3dWComp}A), we see that $\Delta W_2$ from Eq.~\eqref{E:DeltaW2Lin} performs much more robustly (Figure~\ref{Over3dWComp}B). The cost reliably converges toward zero with increasing rates of convergence at larger learning rates.

%%%%%%%%%%%%%%%%%%%%%%%%%%%%%%%%%%%%%%%%%%%%%%%%%%%%%%%%%%%%%%%%%%%%%%%
%%%%%%%%%%%%%%%%%%%%%%%%%%%%%%%%% sub 3.2.3: Linear Approximation %%%%%%%%%%%%%%%%%%%%%%%%%%%%%%%%%
\subsubsection{Learning using the linearized reparameterized rule in a linear model.}\label{S:linearNDdW3}
%Notice that Eq.~\eqref{E:Deltaw2} is rather a complicated learning rule, which involves several inverse calculations. Since $\Delta W$ is associated with $\eta_A$, we can view it as a function of $\eta_A$, If we assume that $\eta_A$ is small around a central point $a$, then we can approximate Eq.~\eqref{E:Deltaw2} by applying Taylor expansion to a linear order at the center 0. 
%We now consider the linearized reparameterized learning rule. The radius of convergence from Eq.~\eqref{E:radius} for this model can be written as
%\begin{equation}\label{E:convRegionRegression}
%\begin{aligned}
%\eta_A < \frac{1}{\left\|\frac{2}{m}\Big( XX^T-\left[I- W\right]YX^T\Big)\right\|},
%\end{aligned}
%\end{equation} 
The linearized, reparameterized update rule from Eq.~\eqref{E:DeltaW3} for this linear model can be written as
\begin{equation}\label{E:DeltaW3Lin}
\begin{aligned}
\Delta W_3%&=\frac{1}{m}\sum_{i=1}^m \Delta W_3^i \\
%&=-\frac{2\eta_A}{m}A^{-1}\big(R-Y)X^TA^{-1}\\
&= -\frac{2\eta_A}{m}\left[I-W\right] (R-Y)X^T\left[I-W\right].\\
\end{aligned}
\end{equation}
This approximated learning rule gives a simpler equation  that is more efficient to compute, but still shows excellent agreement with the reparameterized rule, $\Delta W_2$ (Figure~\ref{Over3dWComp}C, compare to B).  

Note that  Eq.~\eqref{E:DeltaW3Lin} does not require any explicit computation of matrix inverses with the exception of the computation of firing rates $R=[I-W]^{-1}X$. However, since $R$ is left-multiplied by $[I-W]$ in Eq.~\eqref{E:DeltaW3Lin}, we can get rid of this matrix inverse and write $\Delta W_3$ in the form
\begin{equation}\label{E:DeltaALinearAppImprove}
\begin{aligned}
\Delta W_3 
%&=-\frac{2\eta_A}{m}A^{-1}AXX^TA^{-1}-A^{-1}YX^TA^{-1}\\
&= -\frac{2\eta_A}{m}\bigg(XX^T\left[I-W\right]-\left[I-W\right]YX^T\left[I-W\right]\bigg).
\end{aligned}
\end{equation}
Note that the equation $R=[I-W]^{-1}X$ and Eq.~\eqref{E:DeltaALinearAppImprove} are specific to the linear case, $f(z)=z$. When using a nonlinear $f(z)$, fixed point firing rates, $R$,  cannot generally be computed in closed form, but must be approximated by directly simulating Eq.~\eqref{E:drdt} until convergence. 

%%%%%%%%%% %%%%%%%%%%%%% EXTRA %%%%%%%%%%%%%%%%%%%%%%%%%%

\subsubsection{Comparing the direction of updates.}

%%%%%%%%%%%%%%%%%%%%%%%%%%%%%%%%%%%%%%%%%%%%%%%%%%%%%%%%%%%%%%%%%%%%%%%
%%%%%%%%%%%%%%%%%%%%%%%%%% Figure 5 %%%%%%%%%%%%%%%%%%%%%%%%%%%%%%
 \begin{figure}
 \centering{
 \includegraphics[width=6in]{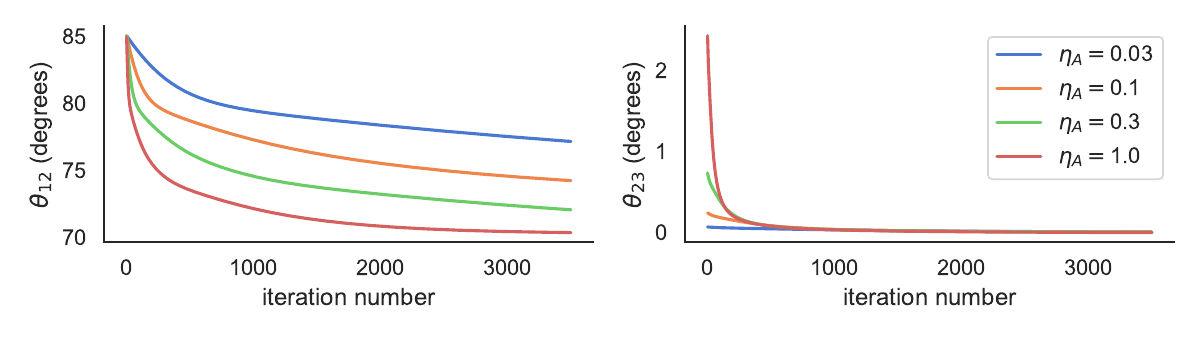}
 }
 \caption{{\bf Angles and correlations between weight updates.} {\bf A)} Angle ($\theta_{12}$) between the weight updates for the gradient-based and reparameterized learning rules. {\bf B)} Same as A, but comparing the parameterized rule with its linearization.   
}
 \label{Angle}
 \end{figure}
%%%%%%%%%%%%%%%%%%%%%%%%%% End of Figure 5

To check the similarity between the updates from each learning rule, we calculated the angle between the updates at each iteration, defined by
%\[
%\Delta W3*\Delta W1 = \eta_W\eta_A\frac{4}{m^2}\left[I-W\right]\left(R-Y\right)X^{T}\left(R-Y\right)R^T
%\]
%\[
%\begin{aligned}
%\theta_{1,2} 
%&= cos^{-1} \left( %\frac{\Delta W_{1}\cdot \Delta W_{2}}{\sqrt{(\Delta W_{1}\cdot \Delta W_{1})(\Delta W_{2}\cdot \Delta W_{2})}}\right)\\
%&= cos^{-1}\left(\frac{Tr(\Delta W_{1}\Delta W_{2}^T)}{\sqrt{Tr(\Delta W_{1}\Delta W_{1}^T)Tr(\Delta W_{2}\Delta W_{2}^T)}}\right).
%\end{aligned}
%\]
\begin{equation*}%\label{E:angle}
\begin{aligned}
\theta_{\alpha\beta} 
&= cos^{-1} \left( \frac{\Delta W_{\alpha}\cdot \Delta W_{\beta}}{\sqrt{(\Delta W_{\alpha}\cdot \Delta W_{\alpha})(\Delta W_{\beta}\cdot \Delta W_{\beta})}}\right)\\
%&= cos^{-1}\left(\frac{Tr(\Delta W_{\alpha}\Delta W_{\beta}^T)}{\sqrt{Tr(\Delta W_{\alpha}\Delta W_{\alpha}^T)Tr(\Delta W_{\beta}\Delta W_{\beta}^T)}}\right)
\end{aligned}
\end{equation*}
for $\alpha,\beta \in\{1,2,3\}$ where $A\cdot B = \textrm{Tr}(A^TB)$ is the Frobenius inner product. 
For sufficiently small learning rates, any update, $\Delta W$, that decreases the cost must satisfy $\Delta W\cdot \Delta W_1>0$ where $\Delta W_1$ is the gradient-based update~\cite{richards2023study} since the change in cost can be written as
\begin{equation}\label{E:DJLin}
\begin{aligned}
\Delta J&= \Delta W\cdot \nabla_W J+\mathcal O(\eta^2)\\
&=-\frac{\eta_W}{m}\Delta W\cdot \Delta W_1+\mathcal O(\eta^2).
\end{aligned}
\end{equation}
Additionally, $\Delta W_2\to \Delta W_3$ as $\eta_A\to 0$. 
Hence, we should expect that $\theta_{\alpha\beta}<90^\circ$ for all pairs, $\alpha$ and $\beta$. 

%The result holds for $\Delta W_3$ in place of $\Delta W_1$ since $\Delta W_3$ is also a gradient-based update, just with a non-Euclidean gradient. 

Figure~\ref{Angle} shows the angles, $\theta_{12}$ and $\theta_{23}$, during learning.  The angles, $\theta_{13}$, were virtually identical to $\theta_{23}$, so they are not shown.
In each example, we used the same update, $\Delta W_\beta$, to update $W$ throughout learning. Hence, the two updates, $dW_\alpha$ and $dW_\beta$, were compared starting at the same initial $W$  at each learning step.  

The angle, $\theta_{12}$, between the gradient-based updates and the reparameterized updates is relatively close to $90^\circ$ (Figure~\ref{Angle}A), indicating that they point in different directions, nearly as different as possible under the condition that they both decrease the cost. 
Unsurprisingly, $\theta_{23}$ is near zero (Figure~\ref{Angle}B), indicating that the reparameterized rule is similar to its linearization.

\subsection{Training fixed points on a nonlinear categorization task.}\label{S:categorial}

So far, for demonstrative purposes, we considered only simple examples of linear regression in which closed equations for optima are known. We next consider an example of image categorization using the MNIST hand-written digit benchmark.  

The learning goal is to minimize a cross-entropy loss on $C = 10$ classes using  one-hot encoded labels. Specifically,
\begin{equation*}%\label{E:CrossEntropy}
L(\vvec{y}^i,\vvec{s}^i) = - {\bf y}^i\cdot  \log(\vvec{s}^i)\\
\end{equation*}
where $\vvec{y}^i$ is a ``one-hot" encoded label for digit $i$, %Specifically, $\vvec{y}^i$ is a $C$-dimensional binary vector with $1$ for the entry corresponding to the correct label and zeros in all other entries.  Also, 
\begin{equation*}%\label{E:SoftMax}
\vvec{s}_l^i = \frac{e^{\vvec{z}^i_l}}{\sum\limits_{k = 1}^{C}e^{\vvec{z}^i_k}}, \\
\end{equation*}
is the softmax output, and $\vvec{z}^i\in \R^{C}$ is a logit computed from a random projection of fixed point rates of a recurrent network. Specifically,
\[
{\bf z}^i=W_\textrm{out}{\bf r}^i
\]
where $W_\textrm{out}\in \R^{C\times N}$ is a fixed, random readout matrix and ${\bf r}^i=f(W{\bf r}^i+{\bf x}^i)$ is the fixed point from an $N\times N$ recurrent network with input $i$. 
Inputs are flattened $28\times 28$ MNIST images, ${\bf p}^i\in \R^M$, where $M=28*28=784$ and we multiply them by a fixed, random read-in matrix to form the input to the network,
\[
{\bf x}^i=W_\textrm{in}{\bf p}^i
\]
where $W_\textrm{in}\in\R^{N\times M}$ and $N=300$ is the number of neurons in the network. We did not train $W_\textrm{out}$ or $W_\textrm{in}$ because we wanted to focus on the effectiveness of learning the recurrent weight matrix, $W$. We used a hyperbolic target activation function, $f(z)=\tanh(z)$. To compute $\bf r$, we simulated Eq.~\eqref{E:drdt} using a forward Euler scheme for $500$ time steps with $\tau=100dt$ where $dt$ is the step size used in the Euler method. The model was trained on $3$ epochs of the MNIST data set using a batch size of $512$. 

 \begin{figure}
 \centering{
 \includegraphics[width=6in]{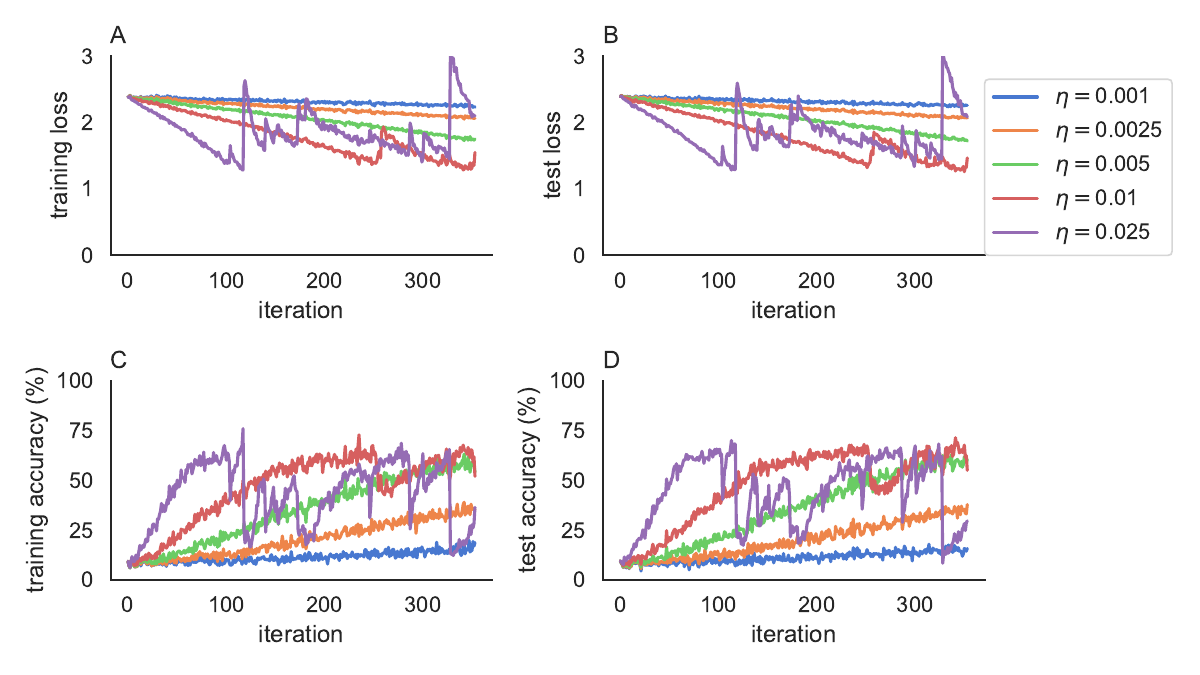}
 }
 \caption{{\bf Gradient based learning on a non-linear classification task.} Results from training the fixed points of a recurrent network to categorize MNIST digits using the gradient-based update rule, $\Delta W_1$.  {\bf A,B,C,D)} Training and testing losses and accuracies evaluated at each step over the course of $3$ epochs. 
}
 \label{MNISTdW1}
 \end{figure}

For this learning task, the gradient-based update rule from Eq.~\eqref{E:DeltaW1c} can be written as
\begin{equation*}%\label{E:DeltaW_CrossEntropyTanh}
\begin{aligned}
\Delta W_1%&=\frac{1}{m}\sum_{i=1}^m \Delta W_1^i\\
&=-\frac{\eta_W}{m}\sum\limits_{i=1}^m\left[\left[G^i\right]^{-1}-W^T\right]^{-1}W_\textrm{out}^T\left[\vvec{s}^i-\vvec{y}^i\right]\left[\vvec{r}^i\right]^T.\\
\end{aligned}
\end{equation*}
We found that this gradient based learning rule performed poorly (Figure~\ref{MNISTdW1}). Small learning rates learned slowly, as expected, while larger learning rates produced instabilities that caused the loss and accuracy to jump erratically during learning.  Indeed, analysis of the Jacobian matrices showed that fixed points became unstable for the two largest learning rates considered in Figure~\ref{MNISTdW1}.

 \begin{figure}
 \centering{
 \includegraphics[width=6in]{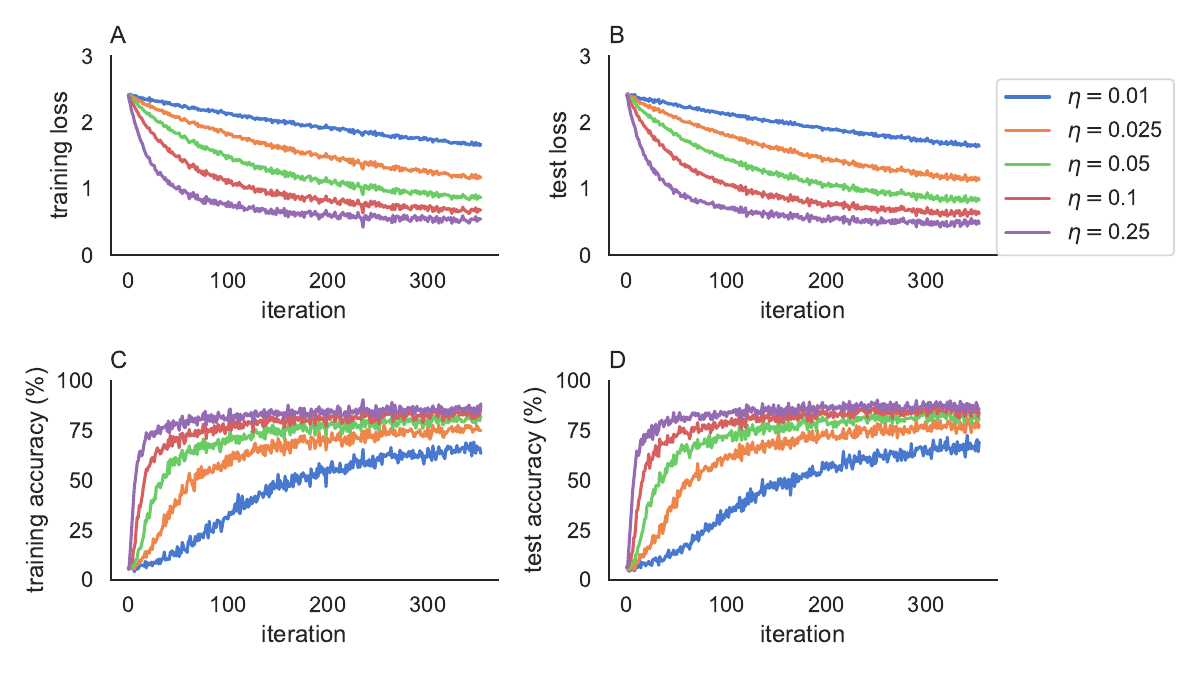}
 }
 \caption{{\bf A reparameterized learning rule on a non-linear classification task.} Results from training the fixed points of a recurrent network to categorize MNIST digits using the gradient-based update rule, $\Delta W_3$.  {\bf A,B,C,D)} Training and testing losses and accuracies evaluated at each step over the course of $3$ epochs. Compare to Figure~\ref{MNISTdW1}. 
}
 \label{MNISTdW3}
 \end{figure}
 
We next tested the linearized, reparameterized update rule, $\Delta W_3$. We did not include results for $\Delta W^2$ because, as in the linear examples considered above, they are very similar to $\Delta W_3$, and they are computationally more expensive to calculate.  For this learning task, $\Delta W_3$ can be written as 
\begin{equation*}%\label{E:DeltaWAlinear_CrossEntropyTanh}
\begin{aligned}
\Delta W_3%&=\frac{1}{m}\sum_{i=1}^m \Delta W_3^i\\
&=-\frac{\eta_A}{m}\sum_{i = 1}^m \left[I-WG^i\right]G^i W_\textrm{out}^T\left[\vvec{s}^i-\vvec{y}^i\right][\vvec{r}^i]^T\left[I-G^iW^T\right]\left[I-G^iW\right].
\end{aligned}
\end{equation*}
Using this linearized, reparameterized update rule significantly improved learning performance (Figure~\ref{MNISTdW3}). Learning performance improved consistently with increasing learning rates and higher accuracy was achieved without instabilities. Analysis of the Jacobian matrices showed that fixed points were stable for all of the learning rates considered in Figure~\ref{MNISTdW3}. 
We conclude that the linearized, reparameterized learning rule can improve the learning of fixed points in non-linear recurrent neural network models.

\section{Discussion}

In summary, we have shown that when learning fixed points of recurrent neural network models, the direct application of gradient descent with respect to the recurrent weight matrix under the Euclidean geometry is computationally expensive and not robust. Badly conditioned loss surfaces  can cause ineffective learning. Moreover, matrix inverses in the equations for the gradients are expensive to evaluate or approximate.

We derived two alternative learning rules derived from a reparameterization of the recurrent network model. These learning rules perform more robustly than the standard gradient descent approach. Moreover, one of the two learning rules is  simpler and more computationally efficient. The learning rules can be interpreted as steepest descent and gradient descent on the recurrent weight matrix under a non-Euclidean metric. Our results support recent calls to re-consider the default use of Euclidean gradients on parameters in machine learning~\cite{amari1998natural,amari1998natural2,martens2020new} and computational neuroscience~\cite{surace2020choice,pogodin2023synaptic}.

Recently, authors have argued that the use of Euclidean gradients for modeling learning in the brain is justified because any learning rule that takes small steps and reduces the loss must be positively correlated with the negative Euclidean gradient~\cite{richards2023study}. Put another way, the angle between the parameter updates and the negative Euclidean gradient must be less than $90^\circ$ (see Eq.~\eqref{E:DJLin} and surrounding discussion). While this is true of the learning rules that we studied, the angle is very close to $90^\circ$ in practice, indicating only a weak correlation. Hence, our work shows that the Euclidean gradient is not always strongly correlated with effective learning rules.

%Much of this work was devoted to deriving the new update rules and understanding them using simple, linear examples. A nonlinear example of categorical learning on the MNIST data set was used to demonstrate that the lessons learning from the linear examples generalize, but even this example is relatively simple. 

We focused on a single, fully connected recurrent layer, which limits the ease with which our model can be applied to larger data sets. Partly for this reason, we only considered the relatively simple MNIST data set as a benchmark. Future work could extend our results to multi-layer recurrent networks in which read-in and read-out matrices are trained and in which at least some fully connected layers are replaced by convolutional connectivity.  These extensions will allow our approach to be applied to larger and more challenging datasets.

Fixed points of recurrent neural networks  are widely used in computational neuroscience to model static neural responses to static stimuli~\cite{ferster2000neural,ozeki2009inhibitory,rubin2015stabilized,ebsch2018imbalanced,curto2019fixed,baker2020nonlinear} and our results could be useful for these modeling approaches. 
On the other hand, recurrent neural networks in machine learning are almost exclusively used for time-varying inputs. Our results rely on the assumption of a time-constant input, ${\bf x}(t)={\bf x}$, which limits their direct application to many machine learning problems. Moreover, even in neuroscience, the assumption of a static stimulus is only an approximation. Natural stimuli are dynamical. However, if fixed points are approached faster than the stimulus changes ({\emph{i.e.}, $\tau$ is faster than ${\bf x}(t)$) then the response, ${\bf r}(t)$, is approximated by the fixed point in Eq.~\eqref{E:FixedPt} and our results provide an approximation. Moreover, a combination of our fixed point learning rules with dynamical learning rules, such as backpropagation through time, could improve learning in situations where some components of the input are static and others are dynamical. Future work should test whether our learning rules can be combined with backpropagation through time to improve performance on tasks with multiple timescales.

\section{Acknowledgments}

This material is based upon work supported by the Air Force Office of Scientific Research (AFOSR) under award number FA9550-21-1-0223 and National Science Foundation (NSF) award numbers DMS-1654268 and DBI-1707400.

%Our results could have implications for learning in the presence of dynamical stimuli as follows: 
%
%Dynamical stimuli in nature can have some components that are slowly varying, or static over behaviorally relevant timescales. For example, consider visual inputs representing moving objects with the task of categorizing the objects. An effective classifier might learn to first 
%
%The stimulus is dynamical, but a non-linear projection onto a 

%%%%%%%%%%% %%%%%%%%%%%%%%%%%%%%%%%%%%%%%%%%%%%%%%
%%%%%%%%%%%%%%%%%%%%%%%%%%%%%%%%%% Appendix Section 
%%%%%%%%%%% %%%%%%%%%%%%%%%%%%%%%%%%%%%%%%%%%%%%%%
\newpage
\appendix

\section{Appendix}

\subsection{Derivation of the Direct Gradient Descent Update, $\Delta W_1$}\label{A:GradDesW1}

Here, we derive Eq.~\eqref{E:DeltaW1b} for direct gradient descent on $W$. 
 To derive Eq.~\eqref{E:DeltaW1b}, it is sufficient to show that 
\[
\nabla_{W} L=\left( \vvec{r} \left[\nabla_{\vvec{r}} L\right]^T \left[I-GW\right]^{-1}G\right)^T.
\]
Here we are considering a single input, label, and fixed point -- $\vvec x$, $\vvec y$, and $\vvec r$ -- so we can omit the $i$ superscripts that appear in Eq.~\eqref{E:DeltaW1b}. 
Note that $\nabla_{W} L(\vvec{r}(W))$ is a matrix with elements
\begin{equation}\label{LWjk}
\frac{\partial L}{\partial W_{jk}} =  [\nabla_{\vvec{r}} L(\vvec{r}, \vvec{y})] \cdot \frac{\partial \vvec{r}}{\partial W_{jk}}.
\end{equation}
To derive $\frac{\partial \vvec{r}}{\partial W_{jk}}$, we first  derive the change of firing rate, $\Delta \vvec r$, to linear order in an update $\Delta W_{jk}$. Consider an initial $\vvec r_0$ satisfying $\vvec r_0=f(W_0\vvec r_0+\vvec x)$ and an update to $W$ defined by $W=W_0+\Delta W$ for some $\Delta W$. The new fixed point satisfies $\vvec r= f(W\vvec r+\vvec x)$ and we wish to compute $\Delta \vvec r=\vvec r-\vvec r_0$ to linear order in $\Delta W$. Define $\vvec z_0=W_0\vvec r_0+\vvec x$ and $\vvec z=W\vvec r+\vvec x$. Then 
\begin{align*}
\Delta \vvec{r} %& =\vvec{r} -\vvec{r}_0\\
	    & = f(\vvec{z})-f(\vvec{z}_0) \hspace{2cm}\\
	    & = f(\vvec{z}_0)+{f'(\vvec{z}_0)}(\vvec{z}-\vvec{z}_0) - f(\vvec{z}_0) + O(\vvec{z}-\vvec{z}_0)^2\\
	    & = G(\vvec{z}-\vvec{z}_0) + O(\vvec{z}-\vvec{z}_0)^2.\\
\end{align*}
To linear order in $\Delta \vvec{r}$, we have
\begin{align*}
\Delta \vvec{r} & = G(\vvec{z}-\vvec{z}_0) \\
	    		& = G \left((W\vvec{r}+\vvec{x}) - (W_0\vvec{r}_0+\vvec{x}_0)\right) %\hspace{1.5cm}\text{for constant input, $\vvec{x} = \vvec{x}_0$}
	    		\\
	    		& = G \left((W_0+\Delta W)\vvec{r}-W_0\vvec{r}_0\right)\\
	    		& = G (W_0\vvec{r}+\Delta W\vvec{r}-W_0\vvec{r}_0)\\
	   		 & = G (W_0\Delta\vvec{r}+\Delta W\vvec{r})\\
\Delta \vvec{r}-GW_0\Delta\vvec{r} & =  G\Delta W\vvec{r}\\
[I-GW_0]\Delta\vvec{r} & =  G\Delta W\vvec{r}.\\
\end{align*}
%Here we see that the changes of firing rates also relies on the changes of synaptic weights.
As a result, we have that
\[
\frac{\partial \vvec{r}}{\partial W_{jk}} = [I - GW]^{-1}G \vvec{1}^{jk}\vvec{r}
\]
which is interpreted as a column vector. Here, $\vvec{1}^{jk}$ is the matrix with all entries equal to zero except for element $(j,k)$, which is equal to $1$. 
Eq.~\eqref{E:DeltaW1b} then follows from the following Lemma.

\begin{lemma}\label{L:drdw}
\begin{equation}\label{LemmaEq}
[I - GW]^{-1}G \vvec{1}^{jk}\vvec{r}=\vvec{r}_k\left[[I - GW]^{-1}G\right]_{(:,j)}
\end{equation}
where $\vvec{r}_k$ is the $k$th element of $\vvec r$ and $B_{(:,j)}$ denotes the $j$th column of a matrix, $B$. 
\end{lemma}
\begin{proof}
We first calculate $\vvec{1}^{11}\vvec{r}$, $\vvec{1}^{12}\vvec{r}$, and $\vvec{1}^{21}\vvec{r}$:
\[
\vvec{1}^{11}\vvec{r} = 
\begin{bmatrix}
1 & 0 &\dots &0\\
0 & \cdot  &\dots &\cdot \\
\cdot  & \cdot  &\dots &\cdot \\
\cdot  & \cdot  &\dots &\cdot \\
0 & 0 &\dots &0\\
\end{bmatrix}
\begin{bmatrix}
\vvec{r}_1\\
\cdot\\
\cdot\\
\cdot\\
\vvec{r}_M\\
\end{bmatrix}=
\begin{bmatrix}
\vvec{r}_1\\
0\\
\cdot\\
\cdot\\
0\\
\end{bmatrix}
= \vvec{r}_1 I(:,1)
\]
\[
\vvec{1}^{12}\vvec{r} = 
\begin{bmatrix}
0 & 1 &\dots &0\\
0 & \cdot  &\dots &\cdot \\
\cdot  & \cdot  &\dots &\cdot \\
\cdot  & \cdot  &\dots &\cdot \\
0 & 0 &\dots &0\\
\end{bmatrix}
\begin{bmatrix}
\vvec{r}_1\\
\cdot\\
\cdot\\
\cdot\\
\vvec{r}_M\\
\end{bmatrix}=
\begin{bmatrix}
\vvec{r}_2\\
0\\
\cdot\\
\cdot\\
0\\
\end{bmatrix}
= \vvec{r}_2 I(:,1)
\]
\[
\vvec{1}^{21}\vvec{r} = 
\begin{bmatrix}
0 & 0 &\dots &0\\
1 & \cdot  &\dots &\cdot \\
\cdot  & \cdot  &\dots &\cdot \\
\cdot  & \cdot  &\dots &\cdot \\
0 & 0 &\dots &0\\
\end{bmatrix}
\begin{bmatrix}
\vvec{r}_1\\
\cdot\\
\cdot\\
\cdot\\
\vvec{r}_M\\
\end{bmatrix}=
\begin{bmatrix}
0\\
\vvec{r}_1\\
\cdot\\
\cdot\\
0\\
\end{bmatrix}
= \vvec{r}_1 I(:,2).
\]
Denote $A: = [I - GW]^{-1}G$ , so $A \vvec{1}^{11}\vvec{r}=\vvec{r}_1A_{(:,1)}$, $A \vvec{1}^{12}\vvec{r}=\vvec{r}_2A_{(:,1)}$, and $A \vvec{1}^{21}\vvec{r}=\vvec{r}_1A_{(:,2)}$. Notice that they are column vectors. WLOG, $A \vvec{1}^{jk}\vvec{r}=\vvec{r}_kA_{(:,j)}$
\[
\begin{aligned}
&LHS  = \nabla_{w} L(\vvec{r}(W)) = 
 				\begin{bmatrix}
				\frac{dL}{dW_{11}} & \frac{dL}{dW_{12}} & \dots & \dots & \frac{dL}{dW_{1M}}\\
				\frac{dL}{dW_{21}} & \frac{dL}{dW_{22}} & \dots & \dots & \frac{dL}{dW_{2M}}\\
				\frac{dL}{dW_{j1}}&\dots & \frac{dL}{dW_{jk}} & \dots & \frac{dL}{dW_{jM}}\\
				\frac{dL}{dW_{M1}} & \frac{dL}{dW_{M2}} & \dots & \dots & \frac{dL}{dW_{MM}}\\
 				\end{bmatrix}\\
 &= 
 			\begin{bmatrix}
			\vvec{r}_1[\nabla_{\vvec{r}} L(\vvec{r})]\cdot A_{(:,1)} & \vvec{r}_2[\nabla_{\text{r}} L(\vvec{r})]\cdot A_{(:,1)} & \dots &\vvec{r}_M[\nabla_{\vvec{r}} L(\vvec{r})]\cdot A_{(:,1)}\\
			\vvec{r}_1[\nabla_{\vvec{r}} L(\vvec{r})]\cdot A_{(:,2)} & \vvec{r}_2[\nabla_{\vvec{r}} L(\vvec{r})]\cdot A_{(:,2)} & \dots &\vvec{r}_M[\nabla_\vvec{r} L(\vvec{r})]\cdot A_{(:,2)}\\
			\dots & \dots & \vvec{r}_k [\nabla_{\vvec{r}} L(\vvec{r})]\cdot A_{(:,j)} & \dots\\
			\vvec{r}_1[\nabla_{\vvec{r}} L(\vvec{r})]\cdot A_{(:,M)} & \vvec{r}_2[\nabla_{\vvec{r}} L(\vvec{r})]\cdot A_{(:,M)} & \dots &\vvec{r}_M[\nabla_\vvec{r} L(\vvec{r})]\cdot A_{(:,M)}\\
			\end{bmatrix},
 \end{aligned}
 \]
 \[
 \begin{aligned}
 &RHS = \left( \vvec{r} [\nabla_{\vvec{r}} L(\vvec{r})]^T [I - GW]^{-1}G \right)^T =  \left(\vvec{r} [\nabla_{\vvec{r}} L(\vvec{r})]^TA\right)^T\\
 \\[4pt]
 &=\left(
  			\begin{bmatrix}
			\vvec{r}_1\frac{\partial L(\vvec{r})}{\partial \vvec{r}_1} & \vvec{r}_1\frac{\partial L(\vvec{r})}{\partial \vvec{r}_2} & \dots &\vvec{r}_1\frac{\partial L(\vvec{r})}{\partial \vvec{r}_M}\\
			\vvec{r}_2\frac{\partial L(\vvec{r})}{\partial \vvec{r}_1} & \vvec{r}_2\frac{\partial L(\vvec{r})}{\partial \vvec{r}_2} & \dots &\vvec{r}_2\frac{\partial L(\vvec{r})}{\partial \vvec{r}_M}\\
			\dots & \dots & \dots & \dots\\
			\vvec{r}_M\frac{\partial L(\vvec{r})}{\partial \vvec{r}_1} & \vvec{r}_M\frac{\partial L(\vvec{r})}{\partial \vvec{r}_2} & \dots &\vvec{r}_M\frac{\partial L(\vvec{r})}{\partial \vvec{r}_M}\\
			\end{bmatrix}
			\begin{bmatrix}
			A_{11} & A_{12}&\dots & A_{1M}\\
			A_{21} & A_{22}&\dots& A_{2M}\\
			\dots &\dots &\dots & \dots\\
			A_{M1} & A_{M2}&\dots & A_{MM}\\
			\end{bmatrix}
			\right)^T\\
\\[4pt]
&=
 			\begin{bmatrix}
			\vvec{r}_1[\nabla_{\vvec{r}} L(\vvec{r})]^T A_{(:,1)} & \vvec{r}_1[\nabla_{\vvec{r}} L(\vvec{r})]^T A_{(:,2)} & \dots & \vvec{r}_1[\nabla_{\vvec{r}} L(\vvec{r})]^T A_{(:,M)}\\
			 \vvec{r}_2[\nabla_{\vvec{r}} L(\vvec{r})]^T A_{(:,1)} & \vvec{r}_2[\nabla_{\vvec{r}} L(\vvec{r})]^T A_{(:,2)} & \dots &\vvec{r}_2[\nabla_{\vvec{r}} L(\vvec{r})]^T A_{(:,M)}\\
			\dots & \dots & \dots & \dots\\
			 \vvec{r}_M[\nabla_{\vvec{r}} L(\vvec{r})]^T A_{(:,1)} & \vvec{r}_M[\nabla_{\vvec{r}} L(\vvec{r})]^T A_{(:,2)} & \dots &\vvec{r}_M[\nabla_{\vvec{r}} L(\vvec{r})]^T A_{(:,M)}\\
			\end{bmatrix}^T\\ 
\\[4pt]
&=
 			\begin{bmatrix}
			\vvec{r}_1[\nabla_{\vvec{r}} L(\vvec{r})]\cdot A_{(:,1)} & \vvec{r}_2[\nabla_{\vvec{r}} L(\vvec{r})]\cdot A_{(:,1)} & \dots &\vvec{r}_M[\nabla_{\vvec{r}} L(\vvec{r})]\cdot A_{(:,1)}\\
			\vvec{r}_1[\nabla_{\vvec{r}} L(\vvec{r})]\cdot A_{(:,2)} & \vvec{r}_2[\nabla_{\vvec{r}} L(\vvec{r})]\cdot A_{(:,2)} & \dots &\vvec{r}_M[\nabla_{\vvec{r}} L(\vvec{r})]\cdot A_{(:,2)}\\
			\dots & \dots & \vvec{r}_k[\nabla_{\vvec{r}} L(\vvec{r})]\cdot A_{(:,j)}  & \dots\\
			\vvec{r}_1[\nabla_{\vvec{r}} L(\vvec{r})]\cdot A_{(:,M)} & \vvec{r}_2[\nabla_{\vvec{r}} L(\vvec{r})]\cdot A_{(:,M)} & \dots &r_M[\nabla_{\vvec{r}} L(\vvec{r})]\cdot A_{(:,M)}\\
			\end{bmatrix}\\
&=LHS.
 \end{aligned}
 \]
\end{proof}
%Notice that $\vvec{r}[\nabla_{\vvec{r}} J(\vvec{r})]^T$ is an $M\times M$ matrix, by the property of transpose, one can rewrite it as $\Big([\nabla_{\vvec{r}} J(\vvec{r})]\vvec{r}^T\Big)^T$
Combining Eq.~\eqref{LWjk} with Eq.~\eqref{LemmaEq} gives
\[
 \begin{aligned}
\nabla_{W} L &=\left( \vvec{r}\left[\nabla_{\vvec{r}} L(\vvec{r})\right]^T \left[I-GW\right]^{-1}G\right)^T\\
%		      &=\bigg(\Big([\nabla_{\vvec{r}} J(\vvec{r})]r^T\Big)^T[G^{-1}-W]^{-1}\bigg)^T\\
 \end{aligned}
\]
which can be simplified to get Eq.~\eqref{E:DeltaW1b} for $\Delta W_1$.
%\end{proof}

%%%%%%%%%%%%%%%%% Proof of "wrong" A reparameterization %%%%%%%%%%%%%%%%% 
\subsection{Analysis of a natural reparameterization and its linear approximation}\label{A:oldA}
We now consider the updates given by the reparameterization $A = \left[G^{-1}-W\right]^{-1}$. The direct reparameterized update, $\Delta W_2$ in this case is given by 
\[
\begin{aligned}
\Delta W_2 & = -\left[\left(A-\eta_A(\nabla_{\vvec{r}}L)(\vvec{x})^{T}\right)^{-1}- A^{-1}\right]\\
	      & = -\left[\left(\left[G^{-1}-W\right]^{-1}-\eta_A(\nabla_{\vvec{r}}L)( \vvec{r})^{T}\left[G^{-1}-W^T\right]\right)^{-1}-G^{-1}-W\right].
\end{aligned}
\]

\begin{proof}
Since $A= \left[G^{-1} - W\right]^{-1}$,  we have $W = G^{-1}-A^{-1}$. Let $W^0$ and $A^0$ represent previous step update before $W$ and $A$, then 
\[
\begin{aligned}
\Delta W & = W-W^0 \\
	      & = G^{-1}-A^{-1}-\left(\left[G^0\right]^{-1}-[A^0]^{-1}\right)\\
	      & = (G^{-1}-\left[G^0\right]^{-1})-A^{-1}+ [A^0]^{-1} \\
	      & = -\left(A^0+\Delta A\right)^{-1}+ [A^0]^{-1} \\
	      & = -\left(A^0-\eta_A\left(\nabla_{\vvec{r}}L\right)\left(\vvec{x}\right)^{T}\right)^{-1}+ [A^0]^{-1} \\
	      & = -\left(A^0-\eta_A\left(\nabla_{\vvec{r}}L\right)\left(\vvec{r}\right)^{T}[A^0]^{-T}\right)^{-1}+ [A^0]^{-1}. \\
\end{aligned}
\]
To get the expression that has only $G$ and $W$, we can substitute $A=[G^{-1}-W]^{-1}$ and $A^{-1} = G^{-1}-W$, and use $G = G^T$ and $G^{-T} = G^{-1}$ since $G$ is a diagonal matrix. This gives
\[
\begin{aligned}
\Delta W_2 & = -\left(A-\eta_A\left(\nabla_{\vvec{r}}L\right)\left(\vvec{r}\right)^{T}A^{-T}\right)^{-1}+ A^{-1} \\
		& =  -\left(\left[G^{-1}-W\right]^{-1}-\eta_A\left(\nabla_{\vvec{r}}L\right)\left(\vvec{r}\right)^{T}\left[G^{-1}-W^{T}\right]\right)^{-1}+ \left[G^{-1}-W\right]. 
\end{aligned}
\]

\end{proof}
Note that as $G_{jj} \to 0$, $A_{jj}^{-1}= [G_{jj}^{i}]^{-1}-W_{jj} \to \infty$, so this reparameterizatin is poorly behaved in situations where $G_{jj}=f'({\bf z}_j)$ becomes small or zero because the second term in the sum diverges while the first term does not. 

We also show that linearizing this parameterization around $\eta_A=0$ still leads to updates that diverge when elements of $G$ become small. 
Following the linearization from Section~\ref{S:dW3}, the linearized, reparameterized update is given by
\[
\begin{aligned}
\Delta W_3 & = -\eta_A A^{-1}(\nabla_{\vvec{r}}L)(\vvec{x})^{T}A^{-1}\\
               & = -\eta_A\left[G^{-1}-W\right](\nabla_{\vvec{r}}L)(\vvec{r})^{T}\left[G^{-1}-W^{T}\right]\left[G^{-1}-W\right].\\
\end{aligned}
\]
\begin{proof} 
First note that $\left.\Delta W_2\right|_{\eta_A=0}=0$, so we have to linear order in $\eta_A$,
\begin{equation}\label{LinW2a}
\Delta W_2 = \left.\frac{d\Delta W_2}{d\eta_A}\right|_{\eta_A=0}\eta_A+\mathcal O(\eta_A^2)
\end{equation}
Now let 
\[
V = A+\Delta A= A-\eta_A(\nabla_{\vvec{r}}L)\left(A^{-1}\vvec{r}\right)^T
\] 
then $\Delta W_2 = A^{-1}-V^{-1}$ so
\[
\begin{aligned}
\frac{d\Delta W_2}{d\eta_A} &= \frac{dA^{-1}}{d\eta_A}-\frac{dV^{-1}}{d\eta_A}\\
&=V^{-1}\frac{dV}{d\eta_A}V^{-1}
\end{aligned}
\]
since $dA^{-1}/d\eta_A=0$. Combining this with Eq.~\eqref{LinW2a} and the definition of $V$ gives the linearized update,
\[
\begin{aligned}
\Delta W_3 &= \left.V^{-1}\frac{dV}{d\eta_A}V^{-1}\right|_{\eta_A=0}\eta_A\\
& =  V^{-1}\left(-\left(\nabla_{\vvec{r}}L\right)\left(A^{-1} \vvec{r}\right)^{T}\right)V^{-1}\Biggr\rvert_{\eta_A = 0} \eta_A\\
%& = -\left(A^{-1}-0\right)(\nabla_{\vvec{r}}L)(A^{-1} \vvec{r})^{T}\left(A^{-1}-0\right)\eta_A\\
& = -A^{-1}(\nabla_{\vvec{r}}L)\left(\vvec{r}\right)^{T}A^{-T}A^{-1}\eta_A\\
& = -\left[G^{-1}-W\right]^{-1}(\nabla_{\vvec{r}}L)\left(\vvec{r}\right)^{T}\left[G^{-1}-W^T\right]\left[G^{-1}-W\right]\eta_A.
\end{aligned}
\]

Again, substitute $A^{-1}=G^{-1}-W$, to get the final expression. Notice that $\Delta W_3 = A^{-1}A^{-T}\Delta W_1 A^{-T}A^{-1}$ so one can let $B = A^{-1}A^{-T}$ and $C = A^{-T}A^{-1}$, which are symmetric, and  $\Delta W_3  = B\Delta W_1 C$.
\end{proof}
Note, again, that $\Delta W_3$ diverges if elements of $G$ go to zero. Therefore, the natural reparameterization $A=[G^{-1}-W]^{-1}$ is not well suited for learning. 

%%%%%%%%%%% Proof of new parameter: linear aprroximation
\subsection{Linearization of the corrected reparameterization}\label{A:LinearApprox}

Here, we derive the linearized update, $\Delta W_3$, given in Eq.~\eqref{E:DeltaW3}. This update rule is derived by expanding $\Delta W_2$ from Eq.~\eqref{E:DeltaW2b} to linear order. Recall that $\Delta W_2$ was derived from the reparameterization $A = [G-GWG]^{-1}$.  
% and it is given by
% \[
% \Delta W_2 = -\left[\left[I-GW\right]^{-1}G-\eta_AG^2\left(\nabla_{\vvec{r}}L\right)(\vvec{r})^T\left[I-GW\right]^TG\right]^{-1}+G^{-1}\left[I-GW\right].\\
% \]
% The linearized approximation update in Eq.~\eqref{E:DeltaW3} is achieved by taking the first order Taylor expansion of the matrix form.
Let $U = [G^{-1} - W]^{-1}$, then we can rewrite Eq.~\eqref{E:DeltaW2b} as
\[
\Delta W_2 = -\left[U-\eta_AG^2\left(\nabla_{\vvec{r}}L\right)(\vvec{r})^T\left[I-GW\right]^TG\right]^{-1}+U^{-1}.\\
\]
Now, denote everything inside of the inverse as $V$ so 
\[V = U-\eta_AG^2\left(\nabla_{\vvec{r}}L\right)(\vvec{r})^T\left[I-GW\right]^TG.
\]
Then Eq.~\ref{E:DeltaW2b} can be further rewritten as 
\[
\Delta W_2 = U^{-1}-V^{-1}.\\
\]
Now, following the same approach as Appendix~\ref{A:oldA}, note that $\left.\Delta W_2\right|_{\eta_A=0}=0$, so the linearization of $\Delta W_2$ around $\eta_A=0$ is given by
\[
\begin{aligned}
\Delta W_3&=\left.\frac{d\Delta W_2}{d\eta_A}\right|_{\eta_A=0}\eta_A\\
&=\left[\frac{d U^{-1}}{d \eta_A}-(-V^{-1}\frac{d V}{d \eta_A}V^{-1}) \right]_{\eta_A = 0}\eta_A\\
& =  V^{-1}\left(-G^2\left(\nabla_{\vvec{r}}L\right)(\vvec{r})^T\left[I-GW\right]^TG\right)V^{-1}\Biggr\rvert_{\eta_A = 0} \eta_A\\
& =  U^{-1}\left(-G^2\left(\nabla_{\vvec{r}}L\right)(\vvec{r})^T\left[I-GW\right]^TG\right)U^{-1} \eta_A\\
&=-\eta_A [G^{-1}-W]G^2\left(\nabla_{\vvec{r}}L\right)(\vvec{r})^T\left[I-GW\right]^TG[G^{-1}-W]\\
&=-\eta_A\left[I-W G\right]G\left(\nabla_{\vvec{r}}L\right)(\vvec{r})^T\left[I - G W\right]^T[I-GW].
%&= U^{-1}\left(-\left(\nabla_{\vvec{r}}L\right)\left(A^{-1} \vvec{r}\right)^{T}\right)U^{-1}\eta_A\\
%& = -U^{-1}(\nabla_{\vvec{r}}L)\vvec{r}^{T}A^{-T}U^{-1}\eta_A\\
%& = -\left[G^{-1}-W\right]^{-1}G^2(\nabla_{\vvec{r}}L)\left(\vvec{r}\right)^{T}\left[I-GW\right]^TG\left[G^{-1}-W\right]\eta_A.
\end{aligned}
\]

\end{document}